\documentclass[aps,pra,superscriptaddress,notitlepage]{revtex4-1}

\usepackage{hyperref}
\usepackage{wrapfig}
\usepackage[margin=1.5cm]{geometry}
\usepackage[titletoc,title]{appendix}
\usepackage{amsmath}
\usepackage{amsfonts}
\usepackage{tabularx}
\usepackage{amsthm}
\usepackage{braket}
\usepackage{listings}
\usepackage{graphicx}
\usepackage{algorithm}
\usepackage{siunitx}
\usepackage{placeins}
\usepackage[noend]{algpseudocode}
\makeatletter
\def\algbackskip{\hskip-\ALG@thistlm}
\makeatother
\renewcommand{\arraystretch}{1.5}

\newtheorem{problem}{Problem}[section]
\newtheorem{lemma}{Lemma}[section]

\usepackage{xcolor,colortbl}
\usepackage{multirow}

\begin{document}

	\title{An Efficient Quantum Compiler that Reduces $T$ Count}
	\date{25\textsuperscript{th} May 2018}
	\author{Luke E. Heyfron}
	\email{leheyfron1@sheffield.ac.uk}
	\affiliation{Department of Physics and Astronomy, University of Sheffield, Sheffield, UK}
	\author{Earl T. Campbell}
	\email{earltcampbell@gmail.com}
	\affiliation{Department of Physics and Astronomy, University of Sheffield, Sheffield, UK}

	\begin{abstract}
	Before executing a quantum algorithm, one must first decompose the algorithm into machine-level instructions compatible with the architecture of the quantum computer, a process known as quantum compiling.
	There are many different quantum circuit decompositions for the same algorithm but it is desirable to compile leaner circuits.
	A fundamentally important cost metric is the $T$ count -- the number of $T$ gates in a circuit.
	For the single qubit case, optimal compiling is essentially a solved problem.
	However, multi-qubit compiling is a harder problem with optimal algorithms requiring classical runtime exponential in the number of qubits.
	Here, we present and compare several efficient quantum compilers for multi-qubit Clifford + $T$ circuits. 
	We implemented our compilers in C++ and benchmarked them on random circuits, from which we determine that our TODD compiler yields the lowest $T$ counts on average.
	We also benchmarked TODD on a library of reversible logic circuits that appear in quantum algorithms and found that it reduced the $T$ count for 97\% of the circuits with an average $T$-count saving of 20\%  when compared against the best of all previous circuit decompositions.
	\end{abstract}
	
	\maketitle
	
		Compiling is the conversion of an algorithm into a series of hardware level commands or elementary gates.
		Better compilers can implement the same algorithm using fewer hardware level instructions, reducing runtime and other resources.
		Quantum compiling or gate-synthesis is the analogous task for a quantum computer and is especially important given the current expense of quantum hardware.
		Early in the field, Solovay and Kitaev proposed a general purpose compiler for any universal set of elementary gates~\cite{1_kitaev_02,2_dawson_05,3_fowler_11}.
		Newer compilers exploit the specific structure of the Clifford+$T$ gate set and have reduced quantum circuit depths by several orders of magnitude~\cite{4_kliuchnikov_13,5_selinger_13,6_gosset_14,7_ross_16},  often improving the classical compile time.
		The Clifford+$T$ gate set is natural since it is the fault-tolerant logical gate set in almost every computing architecture~\cite{8_campbell_16}.
		Moreover, fault-tolerance protocols have been proposed such as magic state distillation~\cite{9_bravyi_05} that lead to a cost per $T$ gate which is several  hundred times larger than that of Clifford gates~\cite{10_raussendorf_07,11_fowler_12,12_gorman_17}, which suggests $T$ count as the key metric of compiler performance.
		Furthermore, the $T$ count is an important metric beyond the standard compiling problem because it relates to the classical overhead of simulating quantum circuits \cite{45_bravyi_2016,46_bravyi_2016,16_howard_17} as well as the distillation cost of synthillation~\cite{24_campbell_17}.
		For these reasons, it is clear that developing methods for minimizing the $T$ count is crucial for a variety of applications in quantum computation.		
		
		Significant progress has been made on synthesis of single-qubit unitaries from Clifford+$T$ gates.
		For purely unitary synthesis, the problem is essentially solved since we have a compiler that is asymptotically optimal and efficient~\cite{4_kliuchnikov_13,7_ross_16}.
		Although further improvements are possible beyond unitary circuits, by making use of ancilla qubits and measurements~\cite{13_paetznick_14,14_bocharov_15,15_bocharov_15b,16_howard_17} or adding an element of randomness to compiling~\cite{17_campbell_17,18_hastings_16}.
		On the other hand, the multi-qubit problem is much more challenging.
		An algorithm for multi-qubit  unitary synthesis over the Clifford+$T$ gate set is known that is provably optimal in terms of the $T$ count but the compile runtime is exponential in the number of qubits~\cite{6_gosset_14,19_Amy_13}.
		Compilers with efficient runtimes have been proposed but with no promise of $T$ count optimality~\cite{20_Amy_13,21_nam_17}.		
		We seek a compiler that runs efficiently and yields circuits with $T$ counts that are as low as practically achievable. 
				
		A useful strategy is to take an initial Clifford+$T$ circuit and split it into subcircuits containing Hadamards and subcircuits containing CNOT, $S$ and $T$ gates.
		One can then attempt to reduce the number of $T$ gates within just the latter type of subcircuit.
		Amy and Mosca recently showed that this restricted problem is formally equivalent to error decoding on a class of Reed-Muller codes~\cite{22_amy_16}, which is in turn equivalent to finding the symmetric tensor rank of a 3-tensor~\cite{23_seroussi_80}.
		Unfortunately, even this easier sub-problem is difficult to solve optimally.
		Nevertheless, it is more amenable to efficient solvers that offer reductions in $T$ count.
		Amy and Mosca proved that an $n$-qubit subcircuit (containing CNOT, $S$ and $T$ gates) has an optimal decomposition into $n^2/2+O(n)$ $T$ gates.
		At the time,  known efficient compilers could only promise an output circuit with no more than $O(n^3) $ $T$ gates.
		Later, Campbell and Howard~\cite{24_campbell_17} sketched a compiler that is efficient and promises an output circuit with at most $n^2/2+O(n)$ $T$ gates.
		This shows efficient compilers can in this sense be ``near-optimal" with respect to worst case scaling.
		On the mathematical level, Campbell and Howard exploited a previously known efficient and optimal solver for a related 2-tensor problem~\cite{25_lempel_75} but suitably modified so that it nearly-optimally solves the required 3-tensor problem.
		
		This paper develops several different compilers that have polynomial runtime in $n$ and are near-optimal in the above sense when restricted to CNOT+T circuits.
		We modify the compiler to also accommodate Hadamard gates using a gadgetisation trick that requires additional resources (measurements, feed-forward and ancillas) and find that it performs well in practice.
		We provide the first implementations of such compilers (the source code is available here \footnote{Source code available at \url{https://github.com/Luke-Heyfron/TOpt}.}) and compare performance against: a family of random circuits; and a library of benchmark circuits that implement actual quantum algorithms.
		For random circuits, we observe $O(n^2)$ scaling in $T$ count for all variants of our compiling approach compared with $O(n^3)$ scaling for compilers based on earlier work.
		Quantum algorithms are highly structured and far from random, so the number of $T$ gates can not be meaningfully compared with the worst case scaling.  
		Instead, we benchmark against the best previously known results and found on average a 20\% $T$ count reduction. 
		In one instance, our compiler gave a 51\% $T$ count reduction and it performed better than previous results for all but one of the benchmarked circuits.
		Of course, the $T$ count is not the only metric relevant to gate synthesis. We discuss the limitations of the $T$ count, as well as other metrics in section~\ref{ssec_T_count}.
		
		All of the near-optimal compilers described in this paper look for inspiration in algorithms for the related 2-tensor problem, which we call Lempel's algorithm. 
		We give specific details for a compiler here called TOOL (Target Optimal by Order Lowering) that comes in two different flavours (with and without feedback). 
		The TOOL compilers can be considered concrete versions of the approach outlined by Campbell and Howard~\cite{24_campbell_17}.
		Also described in this paper is the TODD (Third Order Duplicate and Destroy) compiler, which is again inspired by Lempel but in a more direct and elegant way than TOOL.
		In benchmarking, we find that TODD often achieved even lower $T$ count than TOOL.

	\section{Preliminaries}

	The Pauli group on $n$ qubits $\mathcal{P}^n$ is the set of all $n$-fold tensor products of the single qubit Pauli operators $\{X, Y, Z, \mathbb{I}\}$ with allowed coefficients $\in \{\pm1,\pm i\}$. The $k$\textsuperscript{th} level of the \emph{Clifford hierarchy} $\mathcal{C}_k^n$ is defined as follows,
	\begin{equation}
	\label{e_heir}
	\mathcal{C}_k^n = \{U \mid U\mathcal{P}^n U^\dagger \subseteq \mathcal{C}_{k-1}^n\},
	\end{equation}
	with recursion terminated by $\mathcal{C}^n_1 = \mathcal{P}^n$.
	The Clifford group on $n$ qubits $\mathcal{C}^n$ is the normalizer of $\mathcal{P}^n$.	
	We define $\mathcal{D}_k^n$ to be the diagonal elements of $\langle CNOT, T \rangle$. We will omit the superscript $n$ when the number of qubits is obvious from context.	
	We define Clifford to be any generating set for the Clifford group on $n$ qubits such as $\{CNOT,H,S\}$.
	We define the CNOT + $T$ gate set to be $\{\mathrm{CNOT}, S, T\}$, where we include the phase gate $S=T^2$ as a separate gate due to the magic states cost model for gate synthesis~\cite{9_bravyi_05}.
	A quantum circuit decomposition for a unitary $U$ is denoted $\mathcal{U}$; conversely we say that $\mathcal{U}$ implements $U$.
	Similarly, a circuit $\mathcal{E}$ implements non-unitary channel $\rho \rightarrow \varepsilon(\rho)$. We refer to a circuit $\mathcal{U}$ that implements a $U\in \mathcal{D}_3$ as a \emph{diagonal} CNOT + $T$ \emph{circuit}.

\begin{figure}[h!]
	\includegraphics{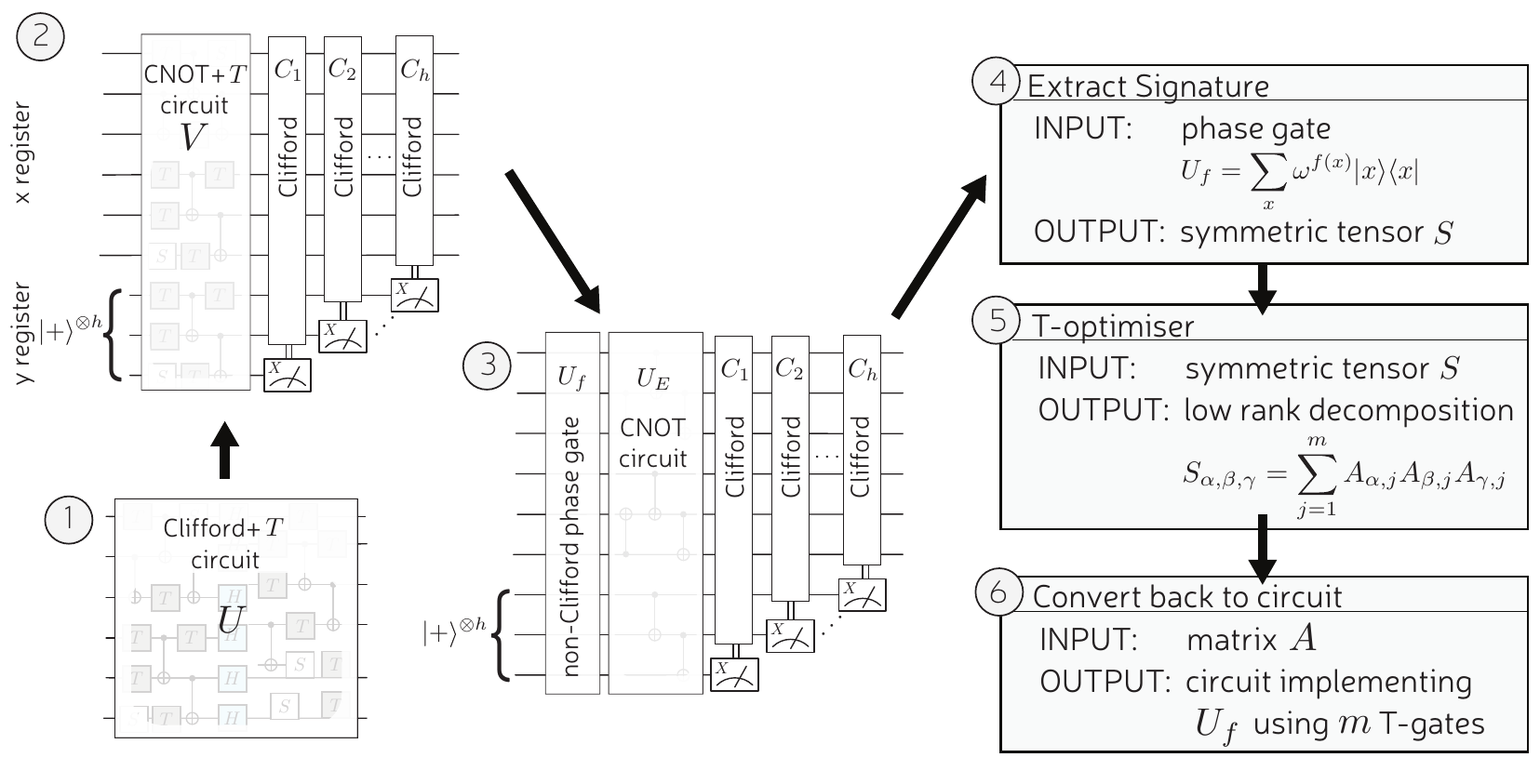}
	\caption{The high level work-flow of the T gate optimization protocol is shown.
	A Clifford + T circuit is converted to the CNOT+T gate set by introducing ancillas and performing classically controlled Clifford gates.
	A non-Clifford phase gate is extracted, which maps to a signature tensor upon which the core optimization algorithm is performed.
	The optimized symmetric tensor decomposition is then converted back into a circuit of the form in panel 2) yielding an implementation of the original Clifford + T circuit with reduced T count. }
	\label{fig_overview}
\end{figure}

\clearpage

\section{Work-flow overview}
		\label{sec_methods}
		\label{ssec_workflow}

In this section, we give a high level work-flow of our approach to compiling as sketched in Fig.~\ref{fig_overview}.
In stages 1-3, some simple circuit preprocessing is performed so that a Clifford+$T$ circuit is converted into a form where the only non-Clifford part is a diagonal CNOT+$T$ gate (an element of $\mathcal{D}_3$). 
Subsection~\ref{ssec_preprocess} describes this preprocessing.
In stages 4-6, the technically difficult aspect of compiling is addressed using a series of different algebraic representations of the circuit and these stages are described in Subsection~\ref{ssec_ciag}.

\subsection{Circuit preprocessing}
\label{ssec_preprocess}

The input circuit $\mathcal{U}_{\text{in}} \in \langle \mathrm{Clifford}, T \rangle$ implements some unitary $U$.
It acts on a register we denote x, which is composed of $n$ qubits and spans the Hilbert space $\mathcal{H}_{\text{x}}$.
The output of our compiler is a circuit $\mathcal{E}_{\text{out}}$ composed of Clifford and $T$ gates but additionally allows: the preparation of $\ket{+}$ states; measurement in the Pauli-X basis, and classical feedforward.
To account for ancilla $\ket{+}$ qubits, we include a register labelled y that is composed of $h$ qubits and spans the Hilbert space $\mathcal{H}_{\text{y}}$.
The circuit $\mathcal{E}_{\text{out}}$ will realise the input unitary after the y register is traced out
\begin{align}
\label{e_out_1}
\mathrm{Tr}_{\text{y}}[\varepsilon_{\text{out}}(\rho_{\text{x}})] & = \mathrm{Tr}_{\text{y}}[\varepsilon_{\text{post}}(V(\rho_{\text{x}}\otimes \ket{+}\bra{+}^{\otimes h}) )V^\dagger)], \\
& = U\rho_{\text{x}}U^\dagger ,
\end{align}
where $\rho_{\text{x}}$ is the density matrix for an arbitrary input pure state on $\mathcal{H}_\text{x}$.
Furthermore,  $V \in \mathcal{C}_3$ is the unitary portion of $\mathcal{E}_{\text{out}}$, and $\varepsilon_{\text{post}}$ is a quantum channel that is associated with the sequence of Pauli-$X$ measurements and subsequent classically controlled Clifford gates, $C_1,C_2,\dots,C_h$, seen in Fig.~\ref{fig_overview}. 

We emphasize that later stages of compiling will make use of a framework valid only for CNOT + $T$ circuits, which makes Hadamard gates an obstacle.
There are two commonly used methods for dealing with Hadamard gates: first, we can partition the quantum circuit into alternating $\langle CNOT, T \rangle$ and $\langle H \rangle$ subcircuits and optimize each CNOT + T subcircuit independently \cite{20_Amy_13}.
The second way is to replace each Hadamard gate with a gadget (see for example references \cite{27_bremner_10,28_montanaro_17}) that makes use of extra resources (ancillas, measurements and feedforward).
The central portion of the gadget contains all of the non-Clifford behaviour and is in the CNOT + T gate set, so is directly compatible with our $T$-optimizers.
The remainder of this section focusses on the second method (Hadamard gadgetization), but we discuss the Hadamard-bounded partitioning method in more detail in appendix~\ref{ap_bench}.

Each\textsuperscript{\footnote{To be precise, gadgets are only need for internal Hadamards.
The external Hadamards that appear at the beginning and end of the circuit do not need to be replaced with Hadamard gadgets.}} of the $h$  Hadamard gates is replaced by a \emph{Hadamard-gadget} (as shown in panel 1) of Fig.~\ref{f_comm}.
A Hadamard-gadget consists of a CNOT + $T$ block followed by a Pauli-$X$ gate conditioned on the outcome of measuring a \emph{Hadamard-ancilla} (a qubit in the y register initialized in the $\ket{+}$ state) in the Pauli-$X$ basis, so the size of the y register is $h$.
After Hadamard-gadgetisation, we commute the classically controlled Pauli-$X$ gates to the end of the circuit, starting with the right-most and iteratively working our way left  (see panel 3 of Fig.~\ref{f_comm}).
The end result is a circuit composed of a single CNOT+$T$ block on $n+h$ qubits, followed by a sequence of classically controlled Clifford operators conditioned on Pauli-$X$ measurements.
The latter sequence of non-unitary gates constitutes the circuit $\mathcal{E}_{\text{post}}$.
This method of circumventing Hadamards is preferred over forming Hadamard-bounded partitions as in previous works \cite{20_Amy_13} because it allows us to convert most of the input circuit into the optimization-compatible gate set, which we find leads to better performance of the \emph{T-Optimiser} subroutine (see appendix \ref{ap_bench} for numerical evidence of this).

\begin{figure}[h!]
	\includegraphics[width=\linewidth]{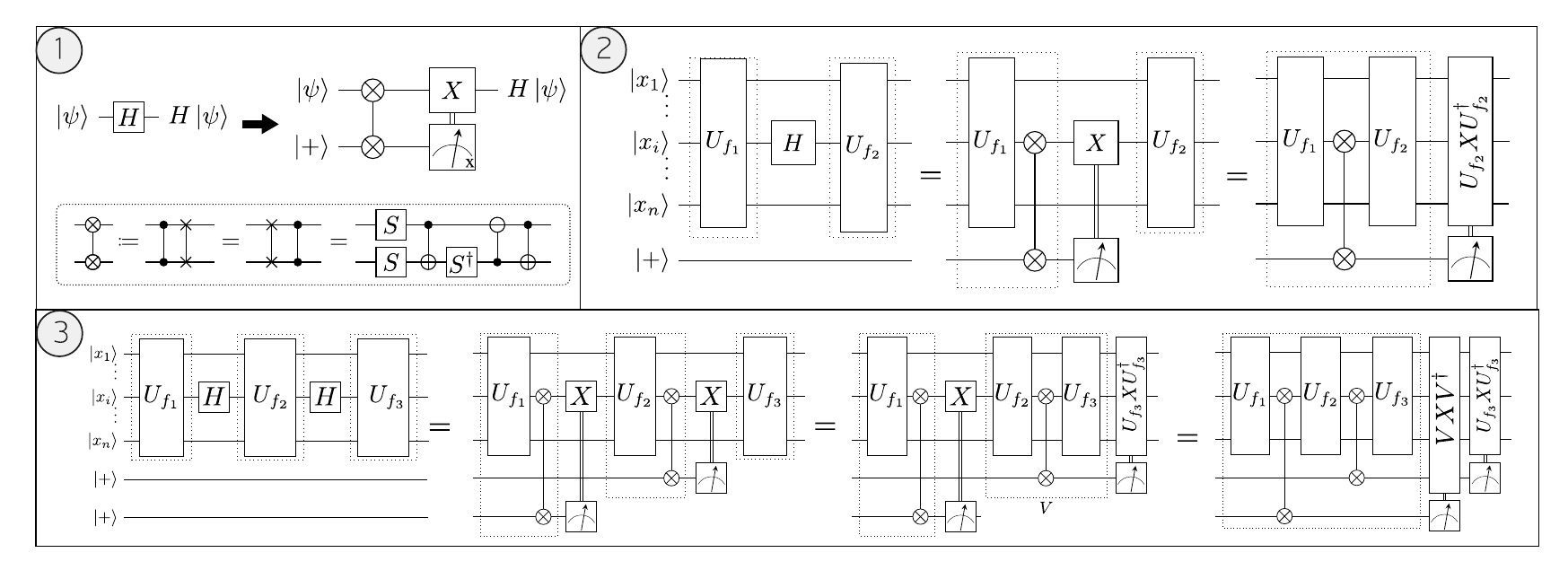}
	\caption{Hadamard gates are replaced by Hadamard-gadgets according to the rewrite in the upper part of panel 1).
			In the lower part, we define notation for the phase-swap gate and provide an example decomposition into the CNOT + $T$ gate set.
			Panel 2) shows an example of a Hadamard gate swapped for a Hadamard-gadget where the classically controlled Pauli-$X$ gate is commuted through $U_{f_2}$ to the end.
			The CNOT + $T$ -only region increases as shown by the dotted lines.
			As $U_{f_2}\in  \mathcal{C}_3$, it follows that $U_{f_2} X U_{f_2}^\dagger \in \mathcal{C}_2$ as per equation~\eqref{e_heir}, so has a $T$-count of $0$.
			The example in panel 3) shows the same process as 2) but for 2 internal Hadamards.
			As $\mathcal{D}_3$ is a group, the operator $V\in\mathcal{D}_3$ and the second Pauli-$X$ gate can also commute to the end to form a Clifford.
			This leads to a decomposition of the form in panel 2) of Fig.~\ref{fig_overview}.}
	\label{f_comm}
\end{figure}

Once the internal Hadamards are removed, we are left with a CNOT + $T$ circuit that implements unitary $V$, whose action on the computational basis is fully described~\cite{20_Amy_13,22_amy_16,30_amy_13,24_campbell_17} by two mathematical objects: a \emph{phase function}, $f: \mathbb{Z}_2^n \mapsto \mathbb{Z}_8$, and an invertible matrix $E \in \mathbb{Z}_2^{(n,n)}$, such that
\begin{equation}
V\ket{\mathbf{x}} = \omega^{f(\mathbf{x})}\ket{E\mathbf{x}} 
\end{equation}
where $\omega = e^{i\frac{\pi}{4}}$. It has been shown~\cite{22_amy_16,24_campbell_17} that $V=U_E U_f$ where  $U_f \in \mathcal{D}_3$ can be implemented with a diagonal CNOT + $T$ circuit and gives the phase 
\begin{equation}
\label{e_U_f_wo}
U_f\ket{\mathbf{x}} = \omega^{f(\mathbf{x})}\ket{\mathbf{x}},
\end{equation}
and $U_E$ can be implemented with CNOTs.

\subsection{Diagonal CNOT+T Framework}
\label{ssec_ciag}

In section \ref{ssec_preprocess}, we isolated all the non-Clifford behaviour of a Clifford + $T$ circuit within a diagonal CNOT + $T$ circuit defined on a larger qubit register. This method allows us to map the $T$ gate optimization problem for any Clifford + $T$ circuit to the following.
\begin{problem}{\textbf{(T-OPT)}}
	\label{p_topt}
	Given a unitary $U_f \in \mathcal{D}_3$, find a circuit decomposition $\mathcal{U}_f \in \langle CNOT, T, S \rangle$ that implements $U_f$ with minimal uses of the $T$ gate.
\end{problem}
\noindent
This section describes how we map the T-OPT problem from the quantum circuit picture to an algebraic problem following stages 4-6 of Fig.~\ref{fig_overview}. 
Throughout this section we use the framework for diagonal CNOT+T circuits (also called \emph{linear phase operators} \cite{22_amy_16}) introduced in reference \cite{30_amy_13} and built upon in \cite{20_Amy_13,22_amy_16, 24_campbell_17}.
We proceed by recalling from equation~\eqref{e_U_f_wo} that the action of any $U_f\in \mathcal{D}_3$ on the computational basis is given by $U_f \ket{\mathbf{x}} = \omega^{f(\mathbf{x})}\ket{\mathbf{x}}$ and that $U_f$ is completely characterized by the phase function, $f$.
A phase function can be decomposed into a sum of linear, quadratic and cubic monomials on the Boolean variables $x_i$.
Each monomial of order $r$ has a coefficient in $\mathbb{Z}_8$ and is weighted by a factor $2^{r-1}$, as in the following:
\begin{equation}
\label{eq_wp}
f(\mathbf{x}) = \sum_{\alpha=1}^{n}l_{\alpha}x_\alpha + 2\sum_{\alpha<\beta}^{n} q_{\alpha,\beta}x_\alpha x_\beta + 4\sum_{\alpha<\beta<\gamma}^{n}c_{\alpha,\beta,\gamma}x_\alpha x_\beta x_\gamma \pmod{8},
\end{equation}
where $l_{\alpha},q_{\alpha,\beta},c_{\alpha,\beta,\gamma} \in \mathbb{Z}_8$.
We refer to decompositions of $f$ that take the form of equation~\eqref{eq_wp} as \emph{weighted polynomials} as in reference \cite{24_campbell_17}, in which 
it was shown that $U_{2f}=U_f^2 \in \mathcal{C}_2$ for any weighted polynomial, $f$. 
This implies that any two unitaries with weighted polynomials whose coefficients all have the same parity are Clifford equivalent.
Note that the weighted polynomial can be lifted directly from the circuit definition of $U_f$ if we work in the $\{T, CS, CCZ\}$ basis, as each kind of gate corresponds to the linear, quadratic and cubic terms, respectively.

In stage 4 of Fig.~\ref{fig_overview}, we define the \emph{signature tensor}, $S^{(U_f)} \in \mathbb{Z}_2^{(n,n,n)}$, to be a symmetric tensor of order 3 whose elements are equal to the parity of the weighted polynomial coefficients of $U_f$ according to the following relations:
\begin{subequations}
	\label{eq_cef_sig}
	\begin{align} 
	S_{\sigma(\alpha,\alpha,\alpha)}& = S_{a,a,a} = l_{\alpha} &\pmod{2} \label{e_cef_sig_1}\\
	S_{\sigma(\alpha,\beta,\beta)}& = S_{\sigma(\alpha,\alpha,\beta)} = q_{\alpha,\beta} &\pmod{2} \label{e_cef_sig_2}\\
	S_{\sigma(\alpha,\beta,\gamma)}& = c_{\alpha,\beta,\gamma} &\pmod{2} \label{e_cef_sig_3}
	\end{align}
\end{subequations}
for all permutations of the indices, denoted $\sigma$.  It follows that any two unitaries with the same signature tensor are Clifford equivalent.

We recall the definition of gate synthesis matrices from reference \cite{24_campbell_17}, where a matrix, $A$ in $\mathbb{Z}_2^{(n,m)}$, is a gate synthesis matrix for a unitary $U_f$ if it satisfies,
\begin{equation}
 \label{eq_gsm}
f(\mathbf{x}) = |A^T\mathbf{x}| \pmod{8} = \sum_j  \left[ \bigoplus_{i} A_{i,j}  x_i \right] \pmod{8}
\end{equation}
where $|.|$ is the Hamming weight of a binary vector.  Notice that inside the square brackets is evaluated modulo 2 and outside is evaluated modulo 8.

Obtaining a gate synthesis matrix from a quantum circuit is best understood via the phase polynomial representation. A phase polynomial of a phase function, $f$, is a set, $P_f=\{\{\lambda_1, a_1\}, \{\lambda_2, a_2\}, \dots, \{\lambda_p,a_{|P|}\}\}$, of linear boolean functions $\lambda_k(\mathbf{x})$, together with coefficients $a_k\in\mathbb{Z}_8$ such that
\begin{equation}
\label{eq_PP}
f(\mathbf{x}) = \sum_{k=1}^{|P_f|}a_k\lambda_k(\mathbf{x}) \pmod{8}.
\end{equation}
A phase polynomial can be extracted  from a diagonal CNOT + $T$ circuit by 
 tracking the action of each gate on the computational basis states through the circuit~\cite{20_Amy_13,30_amy_13}.
We then map $P_f$ to an $A$ matrix with a procedure such as the following.  Start with an empty $A$ matrix. Then for each $\{\lambda_k,a_k\} \in P_f$,
\begin{enumerate}
\item Define column vector, $\mathbf{v}\in\mathbb{Z}_2^n$, such that $\lambda_k(\mathbf{x}) = v_1x_1 \oplus v_2x_2 \oplus \dots \oplus v_nx_n$.
\item Add $a_k$ copies of $\mathbf{v}$ to the right-hand end of $A$.
\end{enumerate}
We define a \emph{proper} gate synthesis matrix to be an $A$ matrix with no all-zero or repeated columns, and we define the function $\textsc{proper}$ such that $A^\prime = \textsc{proper}(A)$ is the proper gate synthesis matrix formed by removing all all-zero columns and pairs of repeated columns from $A$. The purpose of this function is to strip away the Clifford behaviour from the gate synthesis matrix.

We will exploit the key property of $A$ matrices described in the following lemma, which is a corollary of lemma 2 of reference \cite{30_amy_13}.
\begin{lemma}
	\label{l_gsm2circ}
	Let $U_f\in\mathcal{D}_3$ be a unitary with phase function $f(\mathbf{x})=|A^T\mathbf{x}|$ and $A^\prime = \textsc{proper}(A) \in \mathbb{Z}_2^{(n,m)}$. It follows that one can generate a circuit that implements $U_f$ with $m=\mathrm{col}(A')$ uses of the $T$ gate.
\end{lemma}
\begin{proof}
	First, we note from the definition of $A$ in equation~\eqref{eq_gsm} that the $j$\textsuperscript{th} column of $A$ leads to a factor of $\omega^{\lambda_j(\mathbf{x})}$ appearing in the diagonal elements of $U_f$ as written in equation~\eqref{e_U_f_wo}, where $\lambda_j$ is a reversible linear Boolean function given by,
	\begin{equation}
	\lambda_j(\mathbf{x}) = A_{1,j}x_1 \oplus A_{2,j}x_2 \oplus \dots \oplus A_{n,j}x_n. 
	\end{equation}
	The action of a circuit generated by CNOT gates on computational basis state $\ket{\mathbf{x}}$ is to replace the value of each qubit with a reversible linear Boolean function on $x_1, x_2, \dots, x_n$. 
	Next, we show how to add the phase $\omega^{	\lambda_j(\mathbf{x})}$.
	We define $B_j$ to be a CNOT unitary such that after applying $B_j$ the first qubit is mapped $\ket{x_1}\rightarrow \ket{	\lambda_j(\mathbf{x})}$.
	A $T$ gate subsequently applied to this qubit will now produce the desired phase. We then uncompute $B_j$ by reversing the order of the CNOT gates.
	This procedure is repeated for every $j$ until all columns of $A$ have been implemented in this way.
	Only the columns of $A$ that also appear in $A^\prime$ require the use of a $T$ gate as all other columns have duplicates, where any pair of duplicates can be implemented by replacing the $T$ gate with an $S$ gate in the above procedure.
	Therefore the $T$ count is equal to $m=\text{col}(A^\prime)$. 
\end{proof}

The signature tensor of $U_f$ can be determined from an $A$ matrix of $U_f$ using the following relation,
\begin{equation}
	\label{eq_sig}
	S^{(A)}_{\alpha,\beta,\gamma} = \sum_{j=1}^m A_{\alpha,j}A_{\beta,j}A_{\gamma,j} \pmod{2}.
\end{equation}
Therefore, the gate synthesis problem T-OPT reduces to the following tensor rank problem.
\theoremstyle{problem}
\begin{problem}{\textbf{(3-STR)}}
	Given a symmetric tensor of order 3, $S\in \mathbb{Z}_2^{(n,n,n)}$, find a matrix $A \in \mathbb{Z}_2^{(n,m)}$ that satisfies equation~\eqref{eq_sig} with minimal $m$.   
\end{problem}
\noindent
Any algorithm attempting to solve 3-STR can be used in stage 5 of Fig.~\ref{fig_overview}.  The observation that T-OPT reduces to 3-STR is not new as it follows directly from earlier work. Amy and Mosca~\cite{22_amy_16} proved that T-OPT is equivalent to minimum distance decoding of the punctured Reed-Muller code of order $n-4$ and length $n$ (often written as $RM^*(n-4, n)$).  Furthermore, in 1980 Seroussi and Lempel~\cite{23_seroussi_80} recognised that this Reed-Muller decoding problem is equivalent to 3-STR and conjectured that this is a hard computational task.  A non-symmetric generalisation of 3-STR has been proved to be NP-complete~\cite{32_haastad_90}, giving further weight to the conjecture.  This imposes a practical upper bound on the number of qubits, $n_{RM}$, over which circuits can be optimally synthesized.

The problem 3-STR is closely related to
\begin{problem}{\textbf{(2-STR)}}
	Given a symmetric tensor of order 2, $S\in \mathbb{Z}_2^{(n,n)}$, find a matrix $A \in \mathbb{Z}_2^{(n,m)}$ that satisfies
	\begin{equation}
	\label{eq_sig2}
	S^{(A)}_{\alpha,\beta} = \sum_{j=1}^m A_{\alpha,j}A_{\beta,j} \pmod{2}.
	\end{equation}
	 with minimal $m$.  
\end{problem}
\noindent
This could also be stated as a matrix factorisation  $S=AA^T$ problem.
As such, we say any $A$ satisfying $S=AA^T$ is a factor of $S$ and a minimal factor is one with the minimum possible number of columns.
As is often the case in complexity theory, the matrix variant of the problem is considerably simpler than the higher order tensor variant.
Lempel gave an algorithm that finds an optimal solution to 2-STR in polynomial time~\cite{25_lempel_75}.
We call this Lempel's factoring algorithm and for completeness describe it in App.~\ref{ap_lempel}.
Our main strategy to $T$ count optimisation is to take insights from Lempel's algorithm for 2-STR and apply them to 3-STR.
In doing so, our compilers will be efficient but lose the promise of optimality, instead providing approximate solutions to 3-STR and T-OPT.  

In the final stage (see 6 of Fig.~\ref{fig_overview}), we map the output matrix of stage 5 back to a diagonal CNOT + $T$ circuit, $\mathcal{U}_{f^\prime}$, that comprises $m$ instances of the $T$ gate using lemma \ref{l_gsm2circ}.
The circuit $\mathcal{U}_{f^\prime}$ implements a unitary $U_{f^\prime}=U_fU_{\text{Clifford}}$, where $U_{\text{Clifford}}$ is a diagonal Clifford factor.
The input weighted polynomial stored since step 4 contains sufficient information to generate a circuit for $U_{\text{Clifford}}^\dagger$ (see appendix \ref{ap_Cliff}), hence we recover the original unitary, $U_f=U_{f^\prime} U_{\text{Clifford}}^\dagger$.
The final part of step 6 constitutes replacing $\mathcal{U}_f$ with $(\mathcal{U}^\dagger_{\text{Clifford}}\circ\mathcal{U}_{f^\prime})$. At this stage, the protocol terminates returning the final output, $\mathcal{E}_{\text{out}} = ( \mathcal{U}^\dagger_{\text{Clifford}}\circ\mathcal{U}_{f^\prime} \circ \mathcal{U}_E \circ  \mathcal{E}_{\text{post}} )$.

\section{\emph{T-optimiser}}
\label{s_topt}
Until now the \emph{T-optimiser} subroutine of our protocol has been treated as a black box whose input is a signature tensor $S$ and the output is a gate synthesis matrix $A$ with few columns. In this section, we describe the inner workings of the various \emph{T-optimiser}s we have implemented in this work. 

\subsection{Reed-Muller decoder (RM)}
Although Reed-Muller decoding is believed to be hard, a brute force solver can be implemented for a small number of qubits. We implement such a brute force decoder and found its limit to be $n_{\text{RM}}=6$.
To gain some intuition for the complexity of the problem, consider the following. The number of codespace generators for $RM^*(n-4, n)$ is equal to $N_G = \sum_{r=1}^{n-4}{{n}\choose{r}}$.
Therefore, the size of the search space is $N_{\text{search}} = 2^{N_G}$.
On a processor with a clock speed of 3.20GHz, generously assuming we can check one codeword per clock cycle, it would take over $91$ years to exhaustively search this space for $n=7$.
Performing the same back-of-the-envelope calculation for $n=6$, it would take $\approx 7\times 10^{-4}$ seconds.
In practice, we find the brute force decoder executes in around 10 minutes for $n=6$, so the time for $n=7$ would be significantly worse.
Clearly, we need to develop heuristics for this problem.

\subsection{Recursive Expansion (RE)}
The simplest means of efficiently obtaining an $A$ matrix for a given signature tensor $S$ is to make use of the modulo identity $2ab = a + b - a\oplus b$.
More concretely, for each non-zero coefficient in the weighted polynomial $l_\alpha$, $q_{\alpha,\beta}$, $c_{\alpha,\beta,\gamma}$, make the following substitutions to the corresponding monomials:
\begin{align}
x_\alpha &\rightarrow x_\alpha, \\
2x_\alpha x_\beta &\rightarrow x_\alpha + x_\beta - (x_\alpha \oplus x_\beta), \\
4x_\alpha x_\beta x_\gamma &\rightarrow
x_\alpha + x_\beta  + x_\gamma - (x_\alpha \oplus x_\beta) - (x_\alpha\oplus x_\gamma) - (x_\beta\oplus x_\gamma) + (x_\alpha \oplus x_\beta\oplus x_\gamma),
\end{align}
from which the corresponding $A$ matrix can be easily extracted.
We call this the \emph{recursive expansion} (RE) algorithm, which has been shown to yield worst-case $T$ counts of $O(n^3)$.
It is straightforward to understand this cubic scaling because any proper gate synthesis matrix resulting from the RE algorithm may include any column of Hamming weight 3 or less.
There are $\sum_{k=1}^3{{n}\choose{k}}=O(n^3)$ such columns so from lemma \ref{l_gsm2circ} there can be at most $O(n^3)$ $T$ gates in the corresponding circuit decomposition.

\subsection{Target Optimal by Order Lowering (TOOL)}
Campbell and Howard \cite{24_campbell_17} proposed an efficient heuristic for T-OPT that requires at most $O(n^2)$ $T$ gates compared to $O(n^3)$ of the best previous (RE) optimizer.
In the quantum circuit picture, the algorithm involves decomposing the input CNOT + $T$ circuit into a cascade of control-$U_{2\tilde{f}}$ operators where $\tilde{f}$ is quadratic rather than cubic.
Lowering the order in this way means that each control-$U_{2\tilde{f}}$ can be synthesized both efficiently and optimally using Lempel's factoring algorithm.
For this reason we call it the \emph{Target Optimal by Order Lowering} (TOOL) algorithm. Fig.~\ref{Fig_TOOLbasic} shows a single step of how TOOL  pulls out a single control-$U_{2\tilde{f}}$ operator, reducing the number of qubits non-trivially affected by the remaining unitary.
The process is repeated until the circuit is small enough to be solved using the RM algorithm.
The core of the algorithm was already outlined in previous work \cite{24_campbell_17} but for completeness App. \ref{App_TOOL} describes both plain TOOL and a variant called TOOL (with feedback).
This paper presents the first numerical results obtained from an implementation of TOOL.

\begin{figure}
	\includegraphics{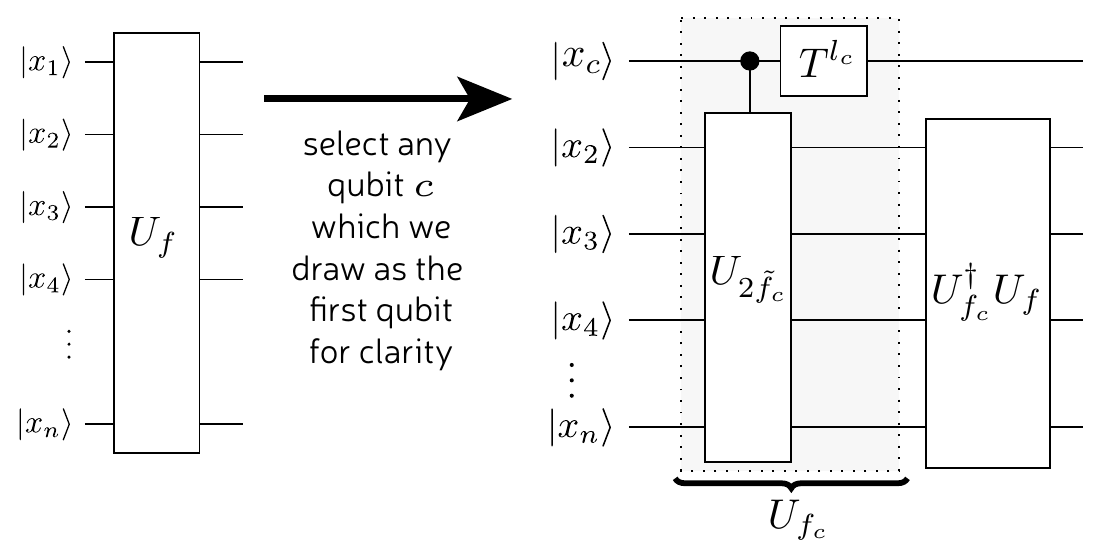}	
	\caption{A sketch of one round of TOOL (without feedback).
			We identify a sub-circuit $U_{f_c}$ with a single control qubit and then use that such a subcirciut can be efficiently and optimally compiled using Lempel's algorithm.
			The remaining circuit $U_{f_c}^\dagger U_f$ contains one fewer qubit and so the process can be iterated until the circuit is down to 6 qubits when it can be optimally compiled by brute force.}
	\label{Fig_TOOLbasic}
\end{figure}

\subsection{Third Order Duplicate and Destroy (TODD)}
\label{sec_TODD}
In this section, we present an algorithm based on Lempel's factoring algorithm \cite{25_lempel_75} that is extended to work for order 3 tensors.
Since this algorithm does not appear in any previous work, we will provide an extended explanation here.
This algorithm requires some initial $A$ matrix to be generated by another algorithm such as RE or TOOL, then it reduces the number of columns of the initial gate synthesis matrix iteratively until exit.
In section \ref{sec_results}, we present numerical evidence that it is the best efficient solver of the T-OPT problem developed so far.
We call this the \emph{Third Order Duplicate and Destroy} (TODD) algorithm because, much like the villainous Victorian barber, it shaves away at the columns of the input $A$ matrix iteratively until the algorithm finishes execution. Pseudo-code is provided in App. \ref{pseudocode}.

We begin by introducing the key mechanism through which TODD reduces the $T$ count of quantum circuits: by \emph{destroying} pairs of duplicate columns of a gate synthesis matrix, a process through which the signature tensor is unchanged, as shown in the following lemma.
\begin{lemma}{}
	\label{lemma_1}
	Let $A\in \mathbb{Z}^{(n,m)}$ be a gate synthesis matrix whose $a$\textsuperscript{th} and $b$\textsuperscript{th} columns are duplicates.
	Let $A_{\mathrm{des}}\in \mathbb{Z}^{(n,m-2)}$ be a gate synthesis matrix formed by removing the $a$\textsuperscript{th} and $b$\textsuperscript{th} columns of $A$.
	It follows that $S^{(A)}=S^{(A_{\mathrm{des}})}$ for any such $A$ and $A_{\mathrm{des}}$.
\end{lemma}
\begin{proof}
	We start by writing the signature tensor in terms of the elements of $A$ according to equation~\eqref{eq_sig},
	\begin{equation}
	S^{(A)}_{\alpha,\beta,\gamma} = \sum_{k=1}^{m}A_{\alpha,k}A_{\beta,k}A_{\gamma,k} \pmod{2},
	\end{equation}
	and separating the terms associated with $a,b$ from the rest of the summation,
	\begin{equation}
	S^{(A)}_{\alpha,\beta,\gamma} = \left( \sum_{j\in \mathcal{J}}A_{\alpha,j}A_{\beta,j}A_{\gamma,j} \right) + A_{\alpha,a}A_{\beta,a}A_{\gamma,a} + A_{\alpha,b}A_{\beta,b}A_{\gamma,b} \pmod{2},
	\end{equation}
	where $\mathcal{J} = \left[1,m\right]\setminus \{a, b\}$, so that
		\begin{equation}
			\label{proof_1_2}
	S^{(A)}_{\alpha,\beta,\gamma} = S^{(A_{\mathrm{des}})}_{\alpha,\beta,\gamma}  + A_{\alpha,a}A_{\beta,a}A_{\gamma,a} + A_{\alpha,b}A_{\beta,b}A_{\gamma,b} \pmod{2},
	\end{equation}
	As stated in the lemma, the $a$\textsuperscript{th} and $b$\textsuperscript{th} columns of $A$ are duplicates and so
	\begin{equation}
		\label{proof_1_3}
		A_{i,a} = A_{i,b}\ \forall\ i \in \left[1,n\right].
	\end{equation}
	Now substitute equation~\eqref{proof_1_3} into equation~\eqref{proof_1_2},
	\begin{align}
	S^{(A)}_{\alpha,\beta,\gamma} & = S^{(A_{\mathrm{des}})}_{\alpha,\beta,\gamma} + 2A_{\alpha,a}A_{\beta,a}A_{\gamma,a}  \pmod{2}  \\
	&= S^{(A_{\mathrm{des}})}_{\alpha,\beta,\gamma} \pmod{2}
	\end{align}
	where the last step follows from modulo 2 addition.
\end{proof}
Lemma \ref{lemma_1} gives us a simple means to remove columns from a gate synthesis matrix by destroying pairs of duplicates columns and thereby reducing the $T$ count of a CNOT + $T$ circuit by 2.
However, it is often the case that the $A$ matrix does not already contain any duplicate columns.
Therefore, we wish to perform some transformation: $A \rightarrow A'$ such that
\begin{enumerate}
	 \item[(\textit{a})] $A'$ has duplicate columns;
	 \item[(\textit{b})] the transformation preserves the signature tensor of $A$.
\end{enumerate}	
In the following lemma we introduce a class of transformations that \emph{duplicate} a particular column of an $A$ matrix such that property (\textit{a}) is met.
We then use lemma \ref{lem1} to establish what conditions must be satisfied for the duplication transformation to have property  (\textit{b}).

\begin{lemma}{}
	\label{lem2}
	Let $A\in\mathbb{Z}_2^{(n,m)}$ be a proper gate synthesis matrix.
	For some choice of $a$ and $b$, let $\mathbf{c}_a(A)$ and $\mathbf{c}_b(A)$ denote the $a$\textsuperscript{th} and $b$\textsuperscript{th} columns of $A$ and define $\mathbf{z}= \mathbf{c}_a(A) \oplus \mathbf{c}_b(A)$ . 
	Let $\mathbf{y} \in \mathbb{Z}_2^m$ be any vector such that $y_a \oplus y_b = 1$.
	We consider duplication transformations of the form $A \rightarrow A^\prime = A \oplus \mathbf{z}\mathbf{y}^T$.
	It follows that the $a$\textsuperscript{th} and $b$\textsuperscript{th} columns of $A^\prime$ are duplicates and so property (\textit{a}) holds.
\end{lemma}
\begin{proof}
	We begin by finding expressions for the matrix elements of $A^\prime$ in terms of $A$, $\mathbf{z}$ and $\mathbf{y}$,
	\begin{equation}
	A^\prime_{i,j} = A_{i,j} \oplus z_i y_j,
	\end{equation}
	and substitute the definition of $\mathbf{z}$,
	\begin{equation}
	A^\prime_{i,j} = A_{i,j} \oplus (A_{i,a}\oplus A_{i,b}) y_j.
	\end{equation}
	Now we can find the elements of the columns $a$ and $b$ of $A^\prime$,
	\begin{align}
	A^\prime_{i,a} &= A_{i,a} \oplus (A_{i,a}\oplus A_{i,b}) y_a,\\
	A^\prime_{i,b} &= A_{i,b} \oplus (A_{i,a}\oplus A_{i,b}) y_b.
	\label{e_working2}
	\end{align}
	We substitute in the condition $y_b = y_a \oplus 1$ into equation~\eqref{e_working2},
	\begin{align}
	\begin{split}
	A^\prime_{i,b} &= A_{i,b} \oplus (A_{i,a}\oplus A_{i,b}) (y_a \oplus 1) \\
	&= A_{i,b} \oplus (A_{i,a} \oplus A_{i,b})y_a \oplus A_{i,a} \oplus A_{i,b} \\
	&= A_{i,a} \oplus (A_{i,a}\oplus A_{i,b})y_a \\
	& = A^\prime_{i,a},
	\end{split}
	\label{eq_lem2_last}		
	\end{align}
	where the two $A_{i,b}$ terms cancel in the second step of equation~\eqref{eq_lem2_last}.
\end{proof}

		\begin{lemma}{}
			\label{lem1}
			Consider a duplication transformation of the form $A \rightarrow A^\prime = A \oplus \mathbf{z}\mathbf{y}^T$ where $\mathbf{z}$, $\mathbf{y}$ are vectors of appropriate length.
			It follows that $S^{(A)} = S^{(A^\prime)}$ (satisfying property (\textit{b})) if the following conditions hold true:
			\begin{enumerate}
				\item [C1:] $\quad |\mathbf{y}| = 0\pmod{2}$
				\item [C2:] $\quad A\mathbf{y} = \mathbf{0}$
				\item [C3:] $\quad \chi(A,\mathbf{z})\hspace{1mm}\mathbf{y} = \mathbf{0}$.
			\end{enumerate}
			where we define $\chi(A,\mathbf{z})$ as follows. Given some gate synthesis matrix, $A$, and a column vector $\mathbf{z}\in\mathbb{Z}_2^n$ let	$\chi$	be a matrix with rows labelled by $(\alpha, \beta, \gamma)$ and of the form
			\begin{equation}
			    \mathbf{R}_{\alpha, \beta, \gamma} =
				(z_ \alpha \mathbf{r_\beta}\wedge\mathbf{r_\gamma})\oplus (z_\beta \mathbf{r_\gamma}\wedge\mathbf{r_\alpha})\oplus (z_\gamma \mathbf{r_\alpha}\wedge\mathbf{r_\beta})
				\label{eq_chi_row}
			\end{equation}
			where $\mathbf{r}_{\alpha}$ is the $\alpha^{\text{th}}$ row of $A$, and $\mathbf{x}\wedge\mathbf{y}$ is the element-wise product of vectors $\mathbf{x}$ and $\mathbf{y}$.
			The order of the rows in $\chi$ is unimportant, but must include every choice of $\alpha, \beta, \gamma \in \mathbb{Z}_n$ with no pair of indices being equal.
		\end{lemma}
		\begin{proof}
			We begin by finding an expression for $S(A^\prime)$ using equation~\eqref{eq_sig},
			\begin{equation}
			S^{(A^\prime)}_{\alpha,\beta,\gamma} = \sum_{j=1}^{m}\left(A_{\alpha,j}\oplus z_\alpha y_j\right)\left(A_{\beta,j}\oplus z_\beta y_j\right)\left(A_{\gamma,j}\oplus z_\gamma y_j\right) \pmod{2},
			\end{equation}
			and expanding the brackets,
			\begin{align}
			\label{e_working1}
			\begin{split}
			S^{(A^\prime)}_{\alpha,\beta,\gamma} = \sum_{j=1}^{m}(&A_{\alpha,j}A_{\beta,j}A_{\gamma,j} \oplus z_\alpha z_\beta z_\gamma y_j  \\			
			&\oplus z_\alpha z_\beta A_{\gamma,j} y_j \oplus z_\beta z_\gamma A_{\alpha,j} y_j \oplus z_\gamma z_\alpha A_{\beta,j} y_j \\
			&\oplus z_\alpha A_{\beta,j} A_{\gamma,j} y_j \oplus z_\beta A_{\gamma,j} A_{\alpha,j} y_j \oplus z_\gamma A_{\alpha,j} A_{\beta,j} y_j) \pmod{2}.
			\end{split}
			\end{align}
			We can see that the first term of equation~\eqref{e_working1} summed over all $j$ is equal to $S^{(A)}$, by definition.
			The task is to show that the remaining terms sum to zero under the specified conditions.
			Next, we sum over all $j$ and substitute in the definitions of $|\mathbf{y}|$, $A\mathbf{y}$ and $\chi(A,\mathbf{z})\hspace{1mm}\mathbf{y}$,
			\begin{equation}
			S^{(A^\prime)}_{\alpha,\beta,\gamma} = S^{(A)}_{\alpha,\beta,\gamma} \oplus z_\alpha z_\beta z_\gamma |\mathbf{y}| \oplus z_\alpha z_\beta \left[A\mathbf{y}\right]_\gamma \oplus z_\beta z_\gamma \left[A\mathbf{y}\right]_\alpha \oplus z_\gamma z_\alpha \left[A\mathbf{y}\right]_\beta \oplus
			(			    \mathbf{R}_{\alpha, \beta, \gamma } \cdot \mathbf{y}).
			\end{equation}					
			By applying condition \emph{C1}, the second term is eliminated; by applying condition \emph{C2}, the next three terms are eliminated, and by applying condition \emph{C3}, the final term is eliminated.
		\end{proof}

		Having shown how to duplicate and destroy columns of a gate synthesis matrix, we are ready to describe the TODD algorithm, presented as pseudo-code in algorithm \ref{al_1}.
		Given an input gate synthesis matrix $A$ with signature tensor $S$, we begin by iterating through all column pairs of $A$ given by indices $a,b$.
		We construct the vector $\mathbf{z}=\mathbf{c}_a \oplus \mathbf{c}_b$ where $\mathbf{c}_j$ is the $j$\textsuperscript{th} column of $A$, as in lemma \ref{lem2}.
		We check to see if the conditions in lemma \ref{lem1} are satisfied for $\mathbf{z}$ by forming the matrix,
		\begin{equation}
		\tilde{A}=\begin{pmatrix}
		A \\ \chi(A,\mathbf{z})
		\end{pmatrix}.
		\end{equation}
		Any vector, $\mathbf{y}$, in the null space of $\tilde{A}$ simultaneously satisfies \emph{C2} and \emph{C3} of lemma \ref{lem1}.
		We scan through the null space basis until we find a $\mathbf{y}$ such that $y_a \oplus y_b = 1$.
		At this stage we know that we can remove at least one column from $A$, depending on the following cases
\begin{enumerate}
	\item[\textit{i}:] 		If $|\mathbf{y}|=0 \pmod{2}$ then condition \emph{C1} is satisfied and we can perform the duplication transformation from lemma \ref{lem1};
	\item[\textit{ii}:]    		If $|\mathbf{y}|=1 \pmod{2}$ then we force \emph{C1} to be satisfied by appending a 1 to $\mathbf{y}$ and an all-zero column to $A$ before applying the duplication transformation.
\end{enumerate}			
Finally, we use the function $\textsc{proper}$ as in App.~\ref{pseudocode} to destroy all duplicate pairs to maximize efficiency.
In case \textit{i}, at least two columns have been removed and in case \textit{ii}  at least one column has been removed \footnote{Other column pairs may be destroyed after the duplication transformation in addition to the $a$\textsuperscript{th} and $b$\textsuperscript{th} columns but only for the latter pair is destruction guaranteed.}.
This reduces the number of columns of $A$ and therefore the $T$ count of $U_f$.
We now start again from the beginning, iterating over columns of the new $A$ matrix.
The algorithm terminates if every column pair has been exhausted without success.

	\section{Results \& Discussion}
	\label{sec_results}
	We implemented our compiler, which we call \emph{TOpt}, in C++ including each variant of \emph{T-Optimiser} described in section~\ref{s_topt}, and tested it on two types of benchmark.
	First, we performed a random benchmark, in which we randomly sampled signature tensors from a uniform probability distribution for a range of $n$ and used them as input for the four versions of \emph{T-optimiser}: RE, TOOL (feedback), TOOL (without feedback) and TODD.
	The results for the random benchmark are shown in Fig.~\ref{fig_random}. 
	Second, we tested the compiler on a library of benchmark circuits taken from Dmitri Maslov's Reversible Logic Synthesis Benchmarks Page \cite{35_Maslov_web}, Matthew Amy's GitHub repository for T-par \cite{39_amy_web} and Nam et al's GitHub repository~\cite{auto_github} for reference~\cite{21_nam_17}.
	These circuits implement useful quantum algorithms including Galois Field multipliers, integer addition, n\textsuperscript{th} prime, Hamming coding functions and the hidden weighted bit functions.	
	The results for the quantum algorithm benchmark are listed in Table \ref{tab_CliffT}. For all benchmarks, the results were obtained on the University of Sheffield's Iceberg HPC cluster\cite{47_iceberg_web}.
	
	\subsection{Random Circuit Benchmark}
	
	We performed the random benchmark in order to determine the average case scaling of the $T$-count with respect to $n$ for each computationally efficient version of \emph{T-optimiser} with results shown in Fig. \ref{fig_random}.
	For both versions of TOOL, we find that the numerical results for the $T$ count follow the expected analytical scaling of $O(n^2)$ and correspondingly the results for RE scales as $O(n^3)$.
	We see that TODD slightly outperforms the next best algorithm, TOOL (without feedback) and is therefore the preferred algorithm in settings where classical runtime is not an issue.
	Furthermore, for all compilers the distribution of $T$-counts (for fixed $n$) concentrates around the mean value.
	Fig.~\ref{fig_random} includes error bars showing the distribution but they are too small to be clearly visible, so for one data point we highlight this with an inset histogram.
	Therefore, TODD performs better, not just on average, but on the vast majority of random circuits so far tested. 
	While both have a polynomial runtime, we found TOOL runs faster than TODD.
	Therefore, TOOL may have some advantage for larger circuits that are impractically large for TODD.
	However, TODD can always partition a very large circuit into several smaller circuits at the cost of being slightly less effective at reducing $T$ count.
	Consequently, for very large circuits, it is unclear which compiler will work best and running both is recommended.
 	
	The random benchmark effectively uses diagonal CNOT + T circuits. This gate set is not universal and therefore is computationally limited.
	However, these circuits are generated by $\{T, CS, CCZ\}$, which all commute.
	This means such circuits lie in the computational complexity class IQP (which stands for \emph{instantaneous quantum  polynomial-time}) that feature in proposals for quantum supremacy experiments \cite{27_bremner_10,37_harrow_17,38_shepherd_09}.
	Low cost designs of IQP circuits provided by our compiler would therefore be an asset for achieving quantum supremacy.
	
	\begin{figure}[h!]
		\centering
		\includegraphics{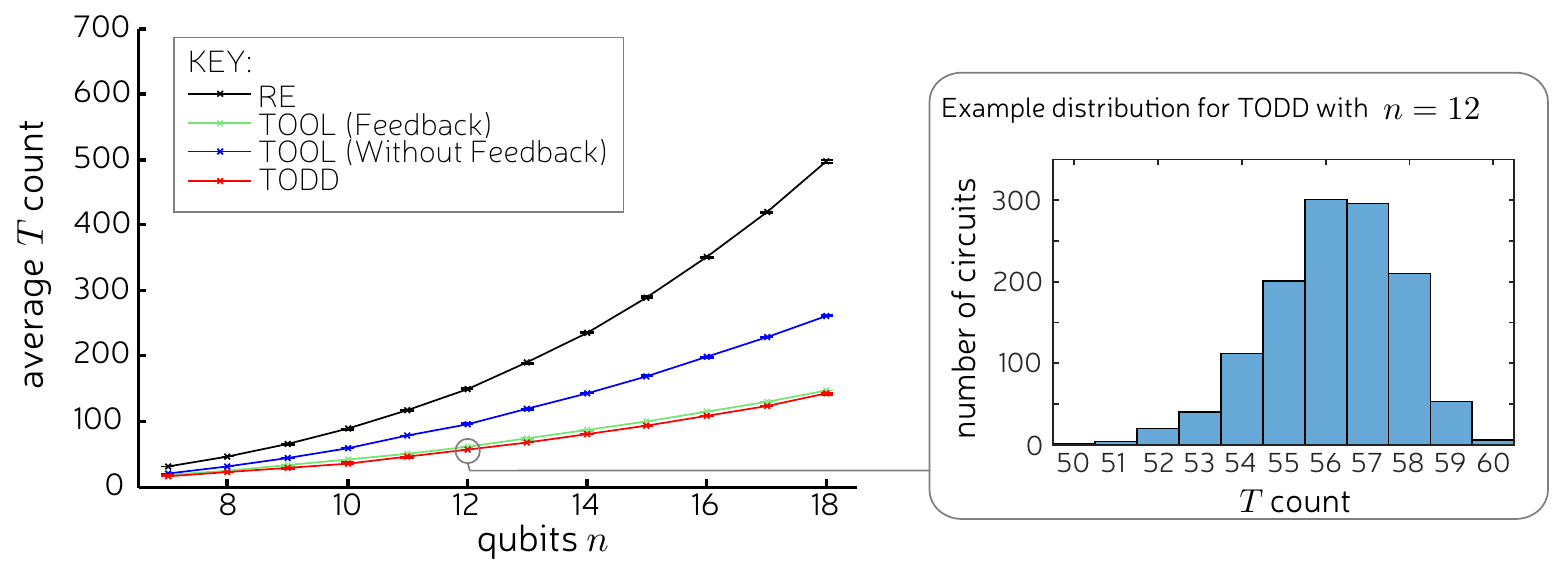}
		\caption{Circuits generated by the $\mathrm{CNOT}$ and $T$ gate were randomly generated for varying number of qubits $n$ then optimized by our implementations of RE, TOOL and TODD.
				The average $T$-count for each $n$ over many random circuits are shown on the vertical axis.
				TODD produces circuit decompositions with the smallest $T$-counts on average but scales the same as the next best algorithm, TOOL (Feedback).
				Both of these algorithms are better than RE by a factor $n$.
				The difference between the $T$-counts for TODD and TOOL (Feedback) seem to converge to a constant $5.5\pm 0.7$ for large $n$.}
		\label{fig_random}
	\end{figure}
	
	\FloatBarrier

	\subsection{Quantum Algorithms Benchmark}
		
	The results in Table \ref{tab_CliffT} show that the TODD algorithm reduced or preserved the $T$ count for every input quantum circuit upon which it was tested, as expected.
	Additionally,  TODD yields a positive saving over the best previous algorithm for all benchmarks except Mod $5_4$ with an average and maximum saving of 20\% and 51\%, respectively.
	This is immediately useful due to the lower cost associated with solving these problems.
	
	Crucially, the output circuits of our protocol often require a considerable number of ancilla qubits due to our use of Hadamard gadgets.
	This space-time trade-off is  justifiable when the cost of introducing an additional qubit is small in comparison to that of performing an additional $T$ gate \cite{44_maslov_16}.
	Furthermore, our compilers can be executed with a cap, $h_{\mathrm{cap}}$, on the size of the ancilla register by dividing the circuit into subcircuits containing no more than $h_{\mathrm{cap}}$  Hadamard gates.	
	A larger number of Hadamard gates generally leads to an increased classical compilation time for TODD as well as an increased $T$ count for TODD-part (see appendix \ref{ap_bench}), which naturally motivates future investigation into Hadamard gate optimization as a pre-processing step of TOpt-like compilers.
	Finally, further reductions in the space (and other) resource requirements may be possible by back-substituting the Hadamard gadget identity from Fig.~\ref{f_comm} post-optimization.
	
	The TOOL algorithms (with and without feedback) reduced $T$ counts below those of the best previous result for 18\% and 30\% of the benchmark circuits, respectively. But for the majority, we find that TOOL actually results in negative savings.
	This seems to contradict the result for the random benchmark (see Fig.~\ref{fig_random}) in which  TOOL (feedback) nearly performs as well as TODD.
	We offer the following explanation for this apparent contradiction.	
	The circuits generated as input for the random benchmark typically have optimal $T$ counts close to the worst-case bound of $O(n^2)$.
	TODD yields $T$ counts very close to optimal because it only terminates when nearly all avenues for $T$ count reduction have been exhausted.
	The TOOL algorithm outputs $T$ counts below $O(n^2)$, so closely competes with TODD for random circuits.
	However, for the Clifford + $T$ benchmark, the optimal $T$ count is typically much less than the worst-case $O(n^2)$ bound. 	
	It is important to recall at this stage that TOOL is optimal for the special case where the circuit implements a control-Clifford.
	But even for this special case, TOOL needs to know which qubit is the control qubit  in order to take advantage of this special case behaviour.
	Consequently, a general-purpose automated compiler without prior knowledge about the input quantum circuit must have access to an additional subroutine which determines the control qubit.
	For general quantum circuits, the task is especially challenging because the circuit must also be optimally partitioned into a sequence of control-Cliffords. As such, we have left this task as an avenue of future work.
	Our implementation of TOOL uses a naive random control-qubit selection subroutine, so regardless of the low optimal $T$ count, TOOL will often output $T$ counts that remain close to the worst-case of $O(n^2)$.
	We suggest that this is the principle cause for the relatively poor performance seen in Table~\ref{tab_CliffT}, which has lead to negative savings not only over the best previous result and TODD, but sometimes also over the input circuit, and conclude that a better control-qubit selector would unlock more of TOOL's $T$-optimizing potential.

	\subsection{The $T$ Count and Other Metrics}
	\label{ssec_T_count}
	
	We acknowledge that the $T$ count does not account for the full space-time cost of quantum computation.
	Recall that we justified neglecting the cost of Clifford gates due to the high ratio between the cost of the $T$ gate and that of Clifford gates.
	The full space-time cost is highly sensitive to the architecture of the quantum computer, but for the surface code, this ratio is estimated to be between 50 and 1000 \cite{40_raussendorf_07,12_gorman_17,42_fowler_13,43_fowler_12}, depending on architectural assumptions.
		
	Note that while our protocol leads to circuits with low $T$ count, the final output often has an \emph{increased} CNOT count.
	This is largely due to step 6 of our protocol where we map the phase polynomial back to a quantum circuit using a naive approach.
	Although $T$ gates cost significantly more than CNOTs  individually, the lower bound on number of CNOT gates required to implement high complexity reversible functions exceeds the upper bound on the number of $T$ gates required by an amount that grows exponentially in $n$ \cite{44_maslov_16}.
	So for large $n$, our focus should turn instead to CNOT optimization.	
	In this paper, we focus exclusively on $T$ count optimization, which is relevant not just to circuit optimization but also to classical simulation runtime \cite{45_bravyi_2016,46_bravyi_2016,16_howard_17} and distillation of magic states ~\cite{24_campbell_17}.
	For this reason, we omit the CNOT count from our benchmark tables and leave the problem of optimizing CNOT count as an avenue for future work.
	
	\section{Conclusions \& Acknowledgements}
	\label{s_c_and_c}
	In this work, we have developed a framework for compiling and optimizing Clifford + $T$ quantum circuits that reduces the $T$ count.
	This scheme maps the quantum circuit problem to an algebraic problem involving order 3 symmetric tensors, for which we have presented an efficient near-optimal solver, and we have reviewed previous methods.
	We implemented our protocol in C++ and used it to obtain $T$ count data for quantum circuit benchmarks.
	Each variant of the compiler has managed to produce quantum circuits for quantum algorithms with lower $T$-counts than any previous attempts known to us.	 However, we find that the TODD compiler with Hadamard gadgets performs the best in practice.
	This lowers the cost of quantum computation and takes us closer to achieving practical universal fault-tolerant quantum computation.

	We acknowledge support by the Engineering and Physical Sciences Research Council (EPSRC) through grant EP/M024261/1.
	We thank Mark Howard and Matthew Amy for valuable discussions, and Dmitri Maslov for comments on the manuscript.
	We thank Quanlong Wang for spotting an error in an earlier draft of the manuscript.	

	\begin{table}[h!]				
			\caption{\small$T$-counts of Clifford + $T$ benchmark circuits for the TODD, TOOL(F) (with feedback) and TOOL(NF) (without feedback) variants of the TOpt compiler are shown. Results for other variants  can be seen in Table~\ref{tab_CliffT_2} of appendix~\ref{ap_bench}.
				Columns $\mathbf{n}$ and $\mathbf{n_h}$  show the number of qubits for the input circuit and the number of Hadamard ancillas, respectively.
				The $T$-count for the circuit is given: before optimization (Pre-Opt.);  after optimization using the best previous algorithm (Best prev.); and post-optimization using our implementation of TODD, TOOL(F) and TOOL(NF).
				The best previous algorithm is given in the \textbf{Alg.} column where: T-par is from \cite{20_Amy_13}; RM\textsubscript{m} and RM\textsubscript{r} are the majority and recursive Reed-Muller optimizers, respectively, both from \cite{22_amy_16}; and Auto\textsubscript{H} is the heavy version of the algorithm from \cite{21_nam_17}.
				We show the $T$-count saving for each TOpt variant over the best previous algorithm in the \textbf{s} columns and the execution time as run on the Iceberg HPC cluster in the \textbf{t} columns.
				Results where the execution time is marked with $^\dagger$ were obtained using an alternative implementation of TODD that is faster but less stable.
				The row \textbf{Positive saving} shows the proportion of the benchmark circuits, as a percentage, for which the corresponding compiler yields a positive saving over the best previous result.
				}
			\footnotesize
			\def\arraystretch{1.2}						
\begin{tabular}{|l|rr|rl|c|rrS|rrS|rrS|}
                        \hline       & \multicolumn{2}{c|}{\textbf{Pre-Opt.}}       & \multicolumn{2}{c|}{\textbf{Best prev.}}       & \textbf{TOpt}        & \multicolumn{3}{c|}{\textbf{TODD}}       & \multicolumn{3}{c|}{\textbf{TOOL(F)}}        & \multicolumn{3}{c|}{\textbf{TOOL(NF)}} \\ 
				\textbf{Circuit}                & $\mathbf{n}$                  & \textbf{T}                   & \textbf{T}                  & \textbf{Alg.}                 & $\mathbf{n_h}$       & \textbf{T}         & \textbf{t (s)}         & \textbf{s(\%)}        & \textbf{T}         & \textbf{t (s)}         &  \textbf{s(\%)}         & \textbf{T}         & \textbf{t (s)}          &  \textbf{s(\%)}  \\ 
				\hline Mod 5$_4$\textsuperscript{\cite{39_amy_web}}                                &5      &28      &16      & T-par                                    &6      &16      &0.04      &0     &19      &0.38      &-18.75      &19      &0.37      &-18.75\\ 
8-bit adder\textsuperscript{\cite{39_amy_web}}      &24      &399      &213      & RM\textsubscript{m}                       &71      &129      &40914.1      &39.43661972      &279      &71886.1      &-30.98591549      &284      &55574.6      &-33.33333333\\ 
CSLA-MUX$_3$\textsuperscript{\cite{auto_github}}      &16      &70      &58      & RM\textsubscript{r}                      &17      &52      &30.41      &10.34482759      &84      &122.95      &-44.82758621      &73      &84.54      &-25.86206897\\ 
CSUM-MUX$_9$\textsuperscript{\cite{auto_github}}      &30      &196      &76      & RM\textsubscript{r}                       &12      &72      &587.21      &5.263157895      &83      &2081.13      &-9.210526316      &104      &340.19      &-36.84210526\\ 
GF($2^4$)-mult\textsuperscript{\cite{39_amy_web}}      &12      &112      &68      & T-par                                   &7      &54      &8.88      &20.58823529      &75      &5.96      &-10.29411765      &75      &2.38      &-10.29411765\\ 
GF($2^5$)-mult\textsuperscript{\cite{39_amy_web}}      &15      &175      &101      & RM\textsubscript{r}                       &9      &87      &66.83      &13.86138614      &109      &17.6      &-7.920792079      &107      &28.27      &-5.940594059\\ 
GF($2^6$)-mult\textsuperscript{\cite{39_amy_web}}      &18      &252      &144      & RM\textsubscript{r}                      &11      &126      &521.86      &12.5      &165      &82.52      &-14.58333333      &157      &60.16      &-9.027777778\\ 
GF($2^7$)-mult\textsuperscript{\cite{39_amy_web}}      &21      &343      &208      & RM\textsubscript{r}                     &13      &189      &2541.4      &9.134615385      &277      &226.4      &-33.17307692      &209      &122.17      &-0.480769231\\ 
GF($2^8$)-mult\textsuperscript{\cite{39_amy_web}}      &24      &448      &237      & RM\textsubscript{r}                     &15      &230      &36335.7      &2.953586498      &370      &379.97      &-56.11814346      &281      &322.83      &-18.56540084\\ 
GF($2^9$)-mult\textsuperscript{\cite{39_amy_web}}                           &27      &567      &301      & RM\textsubscript{r}                     &17      &295      &50671.1      &1.993355482      &454      &1463.02      &-50.83056478      &351      &816.04      &-16.61129568\\ 
GF($2^{10}$)-mult\textsuperscript{\cite{39_amy_web}}                        &30      &700      &410      & T-par                                   &19      &350      &15860.3$^\dagger$      &14.63414634      &550      &7074.29      &-34.14634146      &434      &988.04      &-5.853658537\\ 
GF($2^{16}$)-mult\textsuperscript{\cite{39_amy_web}}                        &48      &1792      &1040      & T-par                                   &31      &\multicolumn{3}{c|}{-}      &1723      &75204.8      &-65.67307692      &1089      &30061.1      &-4.711538462\\ 
Grover$_5$\cite{39_amy_web}                                                 &9      &52      &52      & T-par                                   &23      &44      &17.07      &15.38461538      &106      &110.29      &-103.8461538      &83      &117.39      &-59.61538462\\ 
Hamming$_{15}$ (low)\textsuperscript{\cite{39_amy_web}}                     &17      &161      &97      & T-par                                   &34      &75      &902.69      &22.68041237      &161      &2787      &-65.97938144      &132      &1041.22      &-36.08247423\\ 
Hamming$_{15}$ (med)\textsuperscript{\cite{39_amy_web}}                     &17      &574      &230      & T-par                                    &85      &162      &12410.8$^\dagger$      &29.56521739      &727      &176275      &-216.0869565      &277      &59112.2      &-20.43478261\\ 
HWB$_6$\textsuperscript{\cite{35_Maslov_web}}                               &7      &105      &71      & T-par                                 &24      &51      &55.66      &28.16901408      &189      &140.79      &-166.1971831      &149      &59.24      &-109.8591549\\ 
Mod-Mult$_{55}$\textsuperscript{\cite{39_amy_web}}                                                             &9      &49      &35      & RM\textsubscript{m\&r}      &10      &17      &0.26      &51.42857143      &35      &5.45      &0      &19      &0.92      &45.71428571\\ 
Mod-Red$_{21}$\textsuperscript{\cite{39_amy_web}}                                                              &11      &119      &73      & T-par                                     &17      &55      &25.78      &24.65753425      &68      &40.82      &6.849315068      &71      &19.76      &2.739726027\\ 
n\textsuperscript{th}-prime$_6$\textsuperscript{\cite{35_Maslov_web}}       &9      &567      &400      & RM\textsubscript{m\&r}                   &97      &208      &37348$^\dagger$      &48      &830      &205869      &-107.5      &344      &135165      &14\\ 
QCLA-Adder$_{10}$\textsuperscript{\cite{39_amy_web}}                        &36      &238      &162      & T-par                                     &28      &116      &5496.66      &28.39506173      &167      &7544.58      &-3.086419753      &180      &4560.78      &-11.11111111\\ 
QCLA-Com$_7$\textsuperscript{\cite{39_amy_web}}                             &24      &203      &94      & RM\textsubscript{m}                       &19      &59      &198.55      &37.23404255      &79      &420.95      &15.95744681      &125      &465.41      &-32.9787234\\ 
QCLA-Mod$_7$\textsuperscript{\cite{39_amy_web}}                             &26      &413      &235      & Auto\textsubscript{H}                    &58      &165      &46574.3      &29.78723404      &295      &35249.2      &-25.53191489      &310      &22355.4      &-31.91489362\\ 
QFT$_4$\textsuperscript{\cite{39_amy_web}}                                  &5      &69      &67      & T-par                                    &39      &55      &93.65      &17.91044776      &67      &1602.91      &0      &59      &2756.34      &11.94029851\\ 
RC-Adder$_6$\textsuperscript{\cite{39_amy_web}}                             &14      &77      &47      & RM\textsubscript{m\&r}                     &21      &37      &18.72      &21.27659574      &48      &1238.12      &-2.127659574      &44      &81.77      &6.382978723\\ 
NC Toff$_4$\textsuperscript{\cite{39_amy_web}}                              &5      &21      &15      & T-par                                     &2      &13      &$<10^{-2}$       &13.33333333      &14      &0.02      &6.666666667      &14      &0.01      &6.666666667\\ 
NC Toff$_5$\textsuperscript{\cite{39_amy_web}}                              &7      &35      &23      & T-par                                    &4      &19      &0.06      &17.39130435      &22      &0.24      &4.347826087      &22      &0.12      &4.347826087\\ 
NC Toff$_6$\textsuperscript{\cite{39_amy_web}}                              &9      &49      &31      & T-par                                    &6      &25      &0.4      &19.35483871      &31      &1146.04      &0      &29      &0.67      &6.451612903\\ 
NC Toff$_{10}$\textsuperscript{\cite{39_amy_web}}                           &19      &119      &71      & T-par                                    &16      &55      &44.78      &22.53521127      &65      &1357.98      &8.450704225      &67      &110.44      &5.633802817\\ 
Barenco Toff$_4$\textsuperscript{\cite{39_amy_web}}                         &5      &28      &16      & T-par                                    &3      &14      &$<10^{-2}$       &12.5      &16      &0.02      &0      &16      &0.03      &0\\ 
Barenco Toff$_5$\textsuperscript{\cite{39_amy_web}}                         &7      &56      &28      & T-par                                  &7      &24      &0.45      &14.28571429      &26      &0.88      &7.142857143      &27      &0.56      &3.571428571\\ 
Barenco Toff$_6$\textsuperscript{\cite{39_amy_web}}                         &9      &84      &40      & T-par                                 &11      &34      &1.94      &15      &42      &12.6      &-5      &42      &2.59      &-5\\ 
Barenco Toff$_{10}$\textsuperscript{\cite{39_amy_web}}                      &19      &224      &100      & T-par                                    &31      &84      &460.33      &16      &120      &1938.01      &-20      &122      &1269.03      &-22\\ 
VBE-Adder$_3$\textsuperscript{\cite{39_amy_web}}                            &10      &70      &24      & T-par                              &4      &20      &0.15      &16.66666667      &24      &1639.76      &0      &38      &1.93      &-58.33333333\\ 
				\hline \multicolumn{6}{|l|}{\textbf{Mean}}                            &           &                  &\textbf{19.76}      &      &      &-31.59      &      &      &-14.13\\ 
				\multicolumn{6}{|l|}{\textbf{Standard error}}                            &           &                  &2.12      &      &      &8.87      &      &      &4.69\\ 
				\multicolumn{6}{|l|}{\textbf{Min}}                           &           &                  &0      &      &      &-216.09      &      &      &-109.86\\ 
				\multicolumn{6}{|l|}{\textbf{Max}}               &           &                    &\textbf{51.43}      &      &      &15.96      &      &      &45.71\\ \hline
\hline \multicolumn{6}{|l|}{\textbf{Positive saving (\%)}}                &      \multicolumn{3}{c|}{\textbf{96.88}}      &      \multicolumn{3}{c|}{\textbf{18.18}}      &     \multicolumn{3}{c|}{\textbf{30.30}} \\ \hline
\end{tabular}			
			\label{tab_CliffT}		
		\end{table}

\FloatBarrier

\clearpage

	\begin{appendices}
	\section{Clifford + $T$ Benchmarks for TODD-part and TODD-$h_{\text{cap}}$}
	\label{ap_bench}

	\begin{table}[h!]			
	\centering		
	\caption{\footnotesize$T$-counts of Clifford + $T$ benchmark circuits for the TODD-part and TODD-$h_{\text{cap}}$ variants of TOpt are shown.
		TODD-part uses Hadamard-bounded partitions rather than Hadamard gadgets and ancillas and TODD-$h_{\text{cap}}$ sets a fixed cap, $h_{\text{cap}}$, on the number of Hadamard ancillas available to the compiler.
		Starting at $h_{\text{cap}}=1$, we iteratively incremented the value of $h_{\text{cap}}$ by $1$ until obtaining the first result with a positive $T$-count saving over the best previous algorithm.
		The value of $h_{\text{cap}}$ for which this occured is reported in the $h_{\text{cap}}$ column, and the number of partitions, $T$-count, execution time and percentage saving for this result are detailed by column group TODD-$h_{\text{cap}}$.
		TODD-$h_{\text{cap}}$ results that yield a positive saving for $h_{\text{cap}}=0$ correspond to results for TODD-part and results that require $h_{\text{cap}}=n_h$ Hadamard ancillas correspond to results for TODD.
		As we are strictly interested in intermediate values of $h_{\text{cap}}$, we omit these data and refer the reader to the appropriate result.
		The number of Hadamard partitions is given by the $\mathbf{N_p}$ columns.
		As in Table~\ref{tab_CliffT}, $\mathbf{n}$  is the number of qubits for the input circuit; \textbf{T} are $T$-counts: for the circuit before optimization (Pre-Opt.); due to the best previous algorithm (Best prev.); and post-optimization using variants of our compiler.
		The best previous algorithm is given in the \textbf{Alg.} column where: T-par is from \cite{20_Amy_13}; RM\textsubscript{m} and RM\textsubscript{r} are the majority and recursive Reed-Muller optimizers, respectively, both from \cite{22_amy_16}; and Auto\textsubscript{H} is the heavy version of the algorithm from \cite{21_nam_17}.
		We show the $T$-count saving for each TOpt variant over the best previous algorithm in the \textbf{s} columns and the execution time as run on the Iceberg HPC cluster in the \textbf{t} columns.
		Results where the execution time is marked with $^\dagger$ were obtained using an alternative implementation of TODD that is faster but less stable.
		\textbf{Positive saving} shows the proportion of the benchmark circuits, as a percentage, for which the corresponding compiler yields a positive saving over the best previous result.
		}
	\footnotesize
	\def\arraystretch{1.17}
	\begin{tabular}{|l|rr|rl|rrrS|rrrrS|}
			\hline & \multicolumn{2}{c|}{\textbf{Pre-Opt.}} & \multicolumn{2}{c|}{\textbf{Best prev.}} & \multicolumn{4}{c|}{\textbf{TODD-part}} & \multicolumn{5}{c|}{\textbf{TODD-$h_{\text{cap}}$}} \\ 
		\textbf{Circuit}          & $\mathbf{n}$            & \textbf{T}             & \textbf{T}            & \textbf{Alg.}            & $\mathbf{N_p}$       & \textbf{T}    & \textbf{t (s)}    &  \textbf{s(\%)}     &$h_{\text{cap}}$& $\mathbf{N_p}$       & \textbf{T}    & \textbf{t (s)}    &  \textbf{s(\%)}  \\ 
	\hline	Mod 5$_4$\textsuperscript{\cite{39_amy_web}}                          &5&28&16& T-par                              &7&18&$<10^{-2}$&-12.5&1&4&16&$<10^{-2}$&0\\ 
8-bit adder\textsuperscript{\cite{39_amy_web}}                        &24&399&213& RM\textsubscript{m}               &20&283&12.63&-32.86384977&13&5&212&227.81&0.469483568\\ 
CSLA-MUX$_3$\textsuperscript{\cite{auto_github}}                       &16&70&58& RM\textsubscript{r}                &7&62&0.38&-6.896551724&5&3&54&3.73&6.896551724\\ 
CSUM-MUX$_9$\textsuperscript{\cite{auto_github}}                       &30&196&76& RM\textsubscript{r}                &3&76&20.31&0&4&2&74&36.57&2.631578947\\ 
Cycle ${17}_3$\textsuperscript{\cite{39_amy_web}}&35&4739&1944& RM\textsubscript{m}                 &573&2625&1001.11&-35.0308642&43&15&1939&25507.5$^\dagger$&0.257202\\ 
GF($2^4$)-mult\textsuperscript{\cite{39_amy_web}}                     &12&112&68& T-par                              &3&56&0.55&17.64705882&\emph{0}&\multicolumn{4}{c|}{\emph{See result for  TODD-part}}\\ 
GF($2^5$)-mult\textsuperscript{\cite{39_amy_web}}                     &15&175&101& RM\textsubscript{r}                 &3&90&6.96&10.89108911&\emph{0}&\multicolumn{4}{c|}{\emph{See result for TODD-part}}\\ 
GF($2^6$)-mult\textsuperscript{\cite{39_amy_web}}                     &18&252&144& RM\textsubscript{r}               &3&132&121.16&8.333333333&\emph{0}&\multicolumn{4}{c|}{\emph{See result for  TODD-part}}\\ 
GF($2^7$)-mult\textsuperscript{\cite{39_amy_web}}                     &21&343&208& RM\textsubscript{r}               &3&185&153.75&11.05769231&\emph{0}&\multicolumn{4}{c|}{\emph{See result for  TODD-part}}\\ 
GF($2^8$)-mult\textsuperscript{\cite{39_amy_web}}                     &24&448&237& RM\textsubscript{r}                &3&216&517.63&8.860759494&\emph{0}&\multicolumn{4}{c|}{\emph{See result for  TODD-part}}\\ 
GF($2^9$)-mult\textsuperscript{\cite{39_amy_web}}                     &27&567&301& RM\textsubscript{r}               &3&301&2840.56&0&8&2&295&3212.53&1.993355482\\ 
GF($2^{10}$)-mult\textsuperscript{\cite{39_amy_web}}                  &30&700&410& T-par                             &3&351&23969.1&14.3902439&\emph{0}&\multicolumn{4}{c|}{\emph{See result for  TODD-part}}\\ 
GF($2^{16}$)-mult\textsuperscript{\cite{39_amy_web}}                  &48&1792&1040& T-par                             &3&922&76312.5$^\dagger$&11.34615385&\multicolumn{5}{c|}{-}\\ 
Grover$_5$\cite{39_amy_web}                                           &9&52&52& T-par                              &18&52&0.02&0&5&4&50&0.3&3.846153846\\ 
Hamming$_{15}$ (low)\textsuperscript{\cite{39_amy_web}}               &17&161&97& T-par                             &22&113&0.53&-16.49484536&5&6&93&2.93&4.12371134\\ 
Hamming$_{15}$ (med)\textsuperscript{\cite{39_amy_web}}               &17&574&230& T-par                              &59&322&1.57&-40&11&7&226&58.08&1.739130435\\ 
Hamming$_{15}$ (high)\textsuperscript{\cite{39_amy_web}}              &20&2457&1019& T-par                             &256&1505&16.84&-47.69381747&13&24&1010&595.8&0.883218842\\ 
HWB$_6$\textsuperscript{\cite{35_Maslov_web}}                         &7&105&71& T-par                             &15&82&0.01&-15.49295775&3&6&68&0.13&4.225352113\\ 
HWB$_8$\textsuperscript{\cite{35_Maslov_web}}                         &12&5887&3531& RM\textsubscript{m\&r}                &709&4187&6.53&-18.57830643&9&110&3517&259.14&0.396488247\\ 
Mod-Adder$_{1024}$ \textsuperscript{\cite{39_amy_web}}                &28&1995&1011& T-par                              &234&1165&98.8&-15.23244313&10&27&978&665.5&3.264094955\\ 
Mod-Adder$_{1048576}$\textsuperscript{\cite{39_amy_web}}              &0&0&7298& T-par                                 &2030&9480&89486.5$^\dagger$&-29.89860236&\multicolumn{5}{c|}{-}\\ 
Mod-Mult$_{55}$\textsuperscript{\cite{39_amy_web}}                    &9&49&35& RM\textsubscript{m\&r}              &6&28&0.02&20&\emph{0}&\multicolumn{4}{c|}{\emph{See result for  TODD-part}}\\ 
Mod-Red$_{21}$\textsuperscript{\cite{39_amy_web}}                     &11&119&73& T-par                              &15&85&0.06&-16.43835616&4&5&69&0.59&5.479452055\\ 
n\textsuperscript{th}-prime$_6$\textsuperscript{\cite{35_Maslov_web}} &6&567&400& RM\textsubscript{m\&r}              &63&402&0.17&-0.5&2&29&384&0.98&4\\ 
n\textsuperscript{th}-prime$_8$\textsuperscript{\cite{35_Maslov_web}} &12&6671&4045& RM\textsubscript{m\&r}                 &774&5034&8.4&-24.4499382&12&105&4043&898.98&0.049443758\\ 
QCLA-Adder$_{10}$\textsuperscript{\cite{39_amy_web}}                  &36&238&162& T-par                              &6&184&223.25&-13.58024691&5&3&157&366.1&3.086419753\\ 
QCLA-Com$_7$\textsuperscript{\cite{39_amy_web}}                       &24&203&94& RM\textsubscript{m}                &7&135&11.62&-43.61702128&16&2&81&170.77&13.82978723\\ 
QCLA-Mod$_7$\textsuperscript{\cite{39_amy_web}}                       &26&413&235& Auto\textsubscript{H}              &15&305&34.76&-29.78723404&23&3&221&289.77$^\dagger$&5.957446809\\ 
QFT$_4$\textsuperscript{\cite{39_amy_web}}                            &5&69&67& T-par                               &38&67&$<10^{-2}$&0&2&13&63&0.02&5.970149254\\ 
RC-Adder$_6$\textsuperscript{\cite{39_amy_web}}                       &14&77&47& RM\textsubscript{m\&r}              &13&59&0.11&-25.53191489&6&3&45&0.97&4.255319149\\ 
NC Toff$_3$\textsuperscript{\cite{39_amy_web}}                        &5&21&15& T-par                               &3&15&$<10^{-2}$&0&\emph{2}&$=n_h$&\multicolumn{3}{c|}{\emph{See result for  TODD}}\\ 
NC Toff$_4$\textsuperscript{\cite{39_amy_web}}                        &7&35&23& T-par                               &5&23&$<10^{-2}$&0&\emph{4}&$=n_h$&\multicolumn{3}{c|}{\emph{See result for  TODD}}\\ 
NC Toff$_5$\textsuperscript{\cite{39_amy_web}}                        &9&49&31& T-par                               &7&31&0.01&0&5&2&29&0.2&6.451612903\\ 
NC Toff$_{10}$\textsuperscript{\cite{39_amy_web}}                     &19&119&71& T-par                              &17&71&0.74&0&10&3&69&12.48&2.816901408\\ 
Barenco Toff$_3$\textsuperscript{\cite{39_amy_web}}                   &5&28&16& T-par                               &4&22&$<10^{-2}$&-37.5&2&2&14&$<10^{-2}$&12.5\\ 
Barenco Toff$_4$\textsuperscript{\cite{39_amy_web}}                   &7&56&28& T-par                              &8&38&0.01&-35.71428571&4&2&26&0.06&7.142857143\\ 
Barenco Toff$_5$\textsuperscript{\cite{39_amy_web}}                   &9&84&40& T-par                               &12&54&0.03&-35&6&2&38&0.35&5\\ 
Barenco Toff$_{10}$\textsuperscript{\cite{39_amy_web}}                &19&224&100& T-par                               &32&134&2.27&-34&16&2&98&54.75&2\\ 
VBE-Adder$_3$\textsuperscript{\cite{39_amy_web}}                      &10&70&24& T-par                               &5&36&0.04&-50&\emph{4}&$=n_h$&\multicolumn{3}{c|}{\emph{See result for  TODD}}\\ \hline
		\multicolumn{5}{|l|}{\textbf{Mean}}            &       &         &&-13.19&\textbf{9}&&           &      &4.05\\ 
		\multicolumn{5}{|l|}{\textbf{Standard error}}            &       &         &&3.15&1.65&&           &      &0.64\\ 
		\multicolumn{5}{|l|}{\textbf{Min}}          &                 &         &&-50&1&&           &      &0\\ 
		\multicolumn{5}{|l|}{\textbf{Max}}        &        &         &&20&43&&     &          &13.83\\  \hline
\hline \multicolumn{5}{|l|}{\textbf{Positive saving (\%)}}          &\multicolumn{4}{c|}{\textbf{20.51}}&\multicolumn{5}{c|}{\textbf{96.30}}\\ \hline

	\end{tabular}
	\label{tab_CliffT_2}		
\end{table}

	In order to investigate the relative effectiveness of the Hadamard gadget and Hadamard-bounded partition methods for dealing with Hadamard gates, we repeated the benchmarks from Table~\ref{tab_CliffT} but for the latter method. The results are shown in the TODD-part column group of Table~\ref{tab_CliffT_2}.
	For the Hadamard partition method, we found that the compiler runtime is significantly decreased, making the optimization of larger quantum circuits feasible. 
	However, the performance is worse in terms of raw $T$ count reductions, often leading to higher $T$ counts than the best previous result. 	
	It is important to note that for a given input circuit, the $T$ count is highly sensitive on the choice of Hadamard partitioning, of which, in general, there are many.
	Our implementation does not optimize over Hadamard partitioning choices, so there is potential for developing a more powerful version of TODD-part that makes use of an advanced Hadamard partitioning algorithm, which may lead to greater $T$ count reductions.
	
	The TODD compiler completely gadgetizes each Hadamard gate, whereas the TODD-part compiler completely partitions the circuit into Hadamard-bounded partitions.
	It is possible to interpolate between these two approaches using a parameter $h_{\text{cap}}$ that enforces a cap on the number of available Hadamard ancillas.
	Upon reaching this cap, the compiler synthesises the circuit encountered so far, freeing up the Hadamard ancillas for the subsequent Hadamard partition.
	We have implemented this feature, and in order to quantify the overhead required to see a $T$ count reduction, we ran each benchmark repeatedly, incrementing the value of $h_{\text{cap}}$ until we saw a reduction over the best previous result.
	The results for this experiment are presented in Table~\ref{tab_CliffT_2}. We found that the relationship between $h_{\text{cap}}$ and $T$ count savings is favourable: relatively few Hadamard gadgets are required to see a reduction over the best previous result.
	Over all the benchmark circuits, where the number of qubits and the $T$ count ranges up to $n=36$ and $T=6671$, respectively, we found that on average 9 Hadamard ancillas are required to see positive saving and at most 23 ancillas are needed for all but one exceptional result (Cycle ${17}_3$), which requires 43.
	This suggests that,
	while TODD combined with full Hadamard gadgetization is clearly the forerunner amongst our compilers for reducing the $T$ count,
	a modest improvement in the Hadamard partitioning scheme, or adding a pre-processing step that looks for Hadamard gate reductions may lead to a better version of TODD that requires no non-unitary gadgets, has feasible compiler runtimes for large circuits, and yields positive $T$ count savings.
	
	\FloatBarrier
			
	\section{Lempel's Factoring Algorithm}
	\label{ap_lempel}
	We describe Lempel's factoring algorithm (originally from reference \cite{25_lempel_75}) using conventions consistent with our description of the TODD algorithm to more easily see how TODD generalizes Lempel's algorithm for order 3 tensors. Lempel's factoring algorithm takes as input a symmetric tensor of order 2 (a matrix), which we denote $S\in \mathbb{Z}_2^{(n,n)}$ and outputs a matrix $A\in \mathbb{Z}_2^{(n,m)}$ where the elements of $A$ and $S$ are related as follows:
	\begin{equation}
		\label{eq_lemp1}
		S_{\alpha,\beta} = \sum_{k=1}^{m}A_{\alpha,k}A_{\beta,k} \pmod{2}.
	\end{equation}
Lempel proved that the minimal value of $m$ is equal to
	\begin{equation}
	\mu(S) = \rho(S) + \delta(S),
	\end{equation}
	where $\rho(S)$ is the rank of matrix $S$ and
	\begin{equation}
	\delta(S) = \begin{cases}
	1 & \text{if } S_{\alpha,\alpha}=0\ \forall \ \alpha\in[1,n] \\
	0 & \text{otherwise}
	\end{cases}.
	\end{equation}
	 Lempel's algorithm solves the problem of finding an $A$ matrix that obeys equation~\eqref{eq_lemp1} for a given $S$ matrix such that $m=\mu(S)$. Such an $A$ matrix is referred to as a minimal factor of $S$. 
		
	In the following, we denote the number of columns of $A$ as $c(A)$ and the $j$\textsuperscript{th} column of $A$ as $\mathbf{c}_j(A)$. Lempel's algorithm is the following:
	\begin{enumerate}
		\item Generate an initial (necessarily suboptimal) $A$ matrix for $S$.
		\item Check if $c(A) = \mu(S)$. If true, exit and output $A$. Otherwise, perform steps 3 to 7.
		\item Find a $\mathbf{y}\in\mathbb{Z}_2^m$ such that $A\mathbf{y} = \mathbf{0}$ and $0<|y|<c(A)$.
		\item If $|y|=1 \pmod{2}$ then update $\mathbf{y} \rightarrow (\mathbf{y}^T, 1)^T$ and $A = (A \quad \mathbf{0})$.
		\item Find a pair of indices $a,b\in [1,m], \ a\neq b$ such that $y_a \oplus y_b = 1$.
		\item Apply transformation $A \rightarrow A \oplus \mathbf{z}\mathbf{y}^T$, where $\mathbf{z}=\mathbf{c}_a(A) \oplus \mathbf{c}_b(A)$.
		\item Remove the $a$\textsuperscript{th} and $b$\textsuperscript{th} columns from $A$, then go to step 2.
	\end{enumerate}

	Note that the key difference between the Lempel and TODD algorithm is that TODD additionally requires condition \emph{C3} from lemma~\ref{lem1} to be satisfied.

	\clearpage
	\FloatBarrier
	
	\section{TOOL algorithm}
	\label{App_TOOL}
	 
	 Here we give a detailed description of TOOL, with the main idea illustrated by Fig. \ref{Fig_TOOLbasic}.	TOOL is best explained in terms of weighted polynomials (recall equation~\eqref{eq_wp}). The algorithm is iterative, where each round consists of the five steps detailed below. Before the first round, we initialize 	 
	 an `empty' output gate synthesis matrix, $A_{\text{out}}\in\mathbb{Z}_2^{(n,0)}$.
	\begin{enumerate}
		\item Choose an integer $c \in [1,n]$ such that there is at least one term in $f$ with $x_c$ as a factor. If no such $c$ exists, the algorithm terminates and outputs $A_{\text{out}}$.
		\item Find $\tilde{f}_c$, the \emph{target polynomial} of $f$ with respect to $x_c$ (see equation~\ref{eq_tool1} below).
		\item Determine the order 2 signature tensor, $\tilde{S}$, of $\tilde{f}_c$. 
		\item Find $\tilde{A}$, a minimal factor of $\tilde{S}$, using Lempel's factoring algorithm.
		\item Recover an order 3 gate synthesis matrix, $A$, for $\tilde{A}$, and append it to $A_{\text{out}}$. Replace $f$ with $ f -  |A^T\mathbf{x}|$.\label{tool_final}
	\end{enumerate}
	Each round of TOOL gives a new $f$ that depends on fewer $x$ variables. When $f$ depends on only $n_{\text{RM}}$ or fewer variables, we switch to the optimal brute force optimizer, RM. 
		
	We will now explain each step of the above description in detail, unpacking the contained definitions. 
	In step 1, we select an index $c$, which corresponds to the control qubit of the control-$U_{2\tilde{f}_c}$ operator shown in Fig.~\ref{Fig_TOOLbasic}. The order that we choose $c$ for each round can affect the output and therefore is a parameter of TOOL. For all results, we randomly selected $c$ with uniform probability from the set of all indices $\{c\}$ for which $x_c$ is a factor of at least one term in $f$.
	
	Next, we observe that any $f$ can be decomposed into $f = f_c + f_c^\prime$,  where we define $f_c$ as a weighted polynomial containing all terms of $f$ with $x_c$ as a factor.	
	The former part, $f_c$, can be further decomposed as follows,
	\begin{equation}
	\label{eq_tool1}
	f_c = 2x_c\tilde{f}_c + l_c x_c
	\end{equation}
	where $\tilde{f}_c$ is quadratic and so can be optimally synthesized efficiently. In step 2, we extract $\tilde{f}_c$, which is implicitly fixed by the above equations. We refer to $\tilde{f}_c$ as a \emph{target polynomial} because it corresponds to the target of a control-$U_{2f}$ operator, where $f=\tilde{f}_c$ and $\ket{x_c}$ is the control qubit. 
	
As an aside, we remark that the target polynomial is related to Shannon cofactors that appear in Boole's expansion theorem.  Specifically, we have
		\begin{equation}	
		\tilde{f}_c = \frac{f^+_c-f^-_c-l_c}{2},
		\end{equation}
where $f^+_c$ and $f^-_c$ are the positive and negative Shannon cofactors, respectively, of $f$ with respect to $x_c$, and $l_c$ is the linear coefficient of $f$ associated with $x_c$.	
	
	In step 3, we map $\tilde{f}_c$ to a signature tensor of order 2 (a matrix) for use with Lempel's factoring algorithm. Let $\tilde{l}_{\alpha},  \tilde{q}_{\alpha,\beta}$ be the linear and quadratic coefficients of $\tilde{f}_c$, respectively. For each $\alpha,\beta \neq c$, the elements of $\tilde{S}$ are obtained as follows.
	\begin{equation}
	\tilde{S}_{\alpha,\beta} = \begin{cases}
	\tilde{l}_{\alpha} \pmod{2} & \text{if } \ \alpha=\beta \\
	\tilde{q}_{\alpha,\beta}  \pmod{2} & \text{if } \ \alpha\neq\beta
	\end{cases}.
	\end{equation}
	Finding a minimal factor of $\tilde{S}_{\alpha,\beta}$ is the problem 2-STR.  Therefore,  we can use Lempel's algorithm (see appendix~\ref{ap_lempel}) to find a matrix $\tilde{A}\in\mathbb{Z}_2^{(n,\tilde{m})}$, which is a minimal factor of $\tilde{S}$ such that
	\begin{equation}
	\label{eq_tool2}
	\tilde{f}_c = |\tilde{A}^T\mathbf{x}| = \sum_{j=1}^{\tilde{m}}  \left[ \bigoplus_{i=1}^n \tilde{A}_{i,j}  x_i \right] \pmod{8} .
	\end{equation}
	By substituting equation~\eqref{eq_tool2} into equation~\eqref{eq_tool1} we obtain
	\begin{align}
	\label{eq_f_c}
	f_c & = 2 x_c |\tilde{A}^T\mathbf{x}| + l_c x_c , \\
	& =   \sum_{j=1}^{\tilde{m}} 2x_c \left[ \bigoplus_{i=1}^n   \tilde{A}_{i,j}  x_i \right] + l_c x_c  \pmod{8} , 
	\end{align}
	where we have taken the factor $2x_c$ within the Hamming weight summation.  Next, we use the modular identity $2ab=a+b-a\oplus b$ with $a=x_c$ and $b$ as the contents of the square brackets. This gives 
	\begin{align}	
	f_c & =   \sum_{j=1}^{\tilde{m}} \left( x_c + \left[ \bigoplus_{i=1}^n    \tilde{A}_{i,j}  x_i \right] - x_c  \oplus \left[ \bigoplus_{i=1}^n    \tilde{A}_{i,j}  x_i \right]    \right)  + l_c x_c  \pmod{8} , \\ 
	& =   x_c(\tilde{m}+l_c ) + \sum_{j=1}^{\tilde{m}}  \left[ \bigoplus_{i=1}^n    \tilde{A}_{i,j}  x_i \right] - \sum_{j=1}^{\tilde{m}}  x_c  \oplus \left[ \bigoplus_{i=1}^n    \tilde{A}_{i,j}  x_i \right]      \pmod{8} , \\
	& =   x_c(\tilde{m}+l_c ) + |\tilde{A}^T\mathbf{x}| - |(\tilde{A}\oplus B_c)^T\mathbf{x}|     \pmod{8} ,
	\label{eq_tool_final}
	\end{align}
	where $B_c \in \mathbb{Z}_2^{(n,m)}$ is a matrix with elements
	\begin{equation}
	\left[B_c\right]_{i,j} = \begin{cases}
		1 & \text{if } i=c \\
		0 & \text{otherwise}
	\end{cases}.
	\end{equation}
	This is now in the form of a phase polynomial (e.g. see equation~\eqref{eq_gsm}) with no more than $1+2\tilde{m}$ terms, where $\tilde{m}$ was the optimal size of the factorisation found using Lempel's algorithm.
	
	There are two versions of TOOL: with and without feedback.
	The difference between these versions determines whether all of equation~\eqref{eq_tool_final} is put into $A_{\text{out}}$ or whether parts are `fed back' into $f$ for subsequent rounds. 
	This leads to two distinct definitions of the $A$ matrix referred to in step 5 of TOOL:
	\begin{equation}
	|A^T\mathbf{x}| = \begin{cases}
	(\tilde{m}+l_c)x_c - |(\tilde{A}\oplus B_c)^T\mathbf{x}| & \text{feedback} \\
	(\tilde{m}+l_c)x_c - |(\tilde{A}\oplus B_c)^T\mathbf{x}| + |\tilde{A}^T\mathbf{x}| & \text{without feedback}
	\end{cases}.
	\end{equation}
	Notice that both $(\tilde{m}+l_c)x_c$ and $|(\tilde{A}\oplus B_c)^T\mathbf{x}|$ depend on $x_c$, so must be sent to output.
	Furthermore, they comprise \emph{all} the terms that depend on $x_c$, which is why sending  $|\tilde{A}^T\mathbf{x}|$ to output is optional, and why the number of dependent variables is reduced by at least 1 each round.
	For the \emph{feedback} version, $|\tilde{A}^T\mathbf{x}|$ is kept within $f$ during step 5, whereas it is sent to output $A_{\text{out}}$ in the \emph{without feedback} version. 
	
	\section{Calculating Clifford correction}
	\label{ap_Cliff}
	We will now describe how to determine the Clifford correction required to restore the output of \emph{T-Optimiser} to the input unitary. Let the input of \emph{T-Optimiser} be a weighted polynomial $f$ that implements unitary $U_f \in \mathcal{D}_3$, and let the output be a weighted polynomial $g$. Any $f$ can be split into the sum
	\begin{equation}
	\label{eq_apB1}
	f = f_1 + f_2,
	\end{equation}
	where the coefficients of $f_1$ are in $\mathbb{Z}_2$ and those of $f_2$ are even. From the definition of $\emph{T-Optimiser}$, we know the coefficients of $f$ and $g$ have the same parity i.e.
	\begin{equation}
	\label{eq_apB2}
	g = g_1 + g_2 = f_1 + g_2,
	\end{equation}
	where $g_1, g_2$ are similarly defined for $g$. Using equations~\eqref{eq_apB1} and \eqref{eq_apB2} we find,
	\begin{equation}
	\label{eq_cliff1}
	g = f + (g_2 - f_2).
	\end{equation}
	Equation~\eqref{eq_cliff1} implies that $U_{\text{Clifford}} = U_{(g_2 - f_2)} \in \mathcal{D}_2$. Therefore, the Clifford correction is $U_{\text{Clifford}}^\dagger = U_{(g_2 - f_2)}^\dagger = U_{(f_2-g_2)}$. We can map $(f_2-g_2)$ to a phase polynomial and subsequently to a quantum circuit, $\mathcal{U}_{\text{Clifford}}^\dagger$.

	\clearpage
	\FloatBarrier

	\section{TODD pseudocode}
	\label{pseudocode}
	
		\begin{algorithm}[H]		
		\caption{Third Order Duplicate-then-Destroy (TODD) Algorithm}
		\label{al_1}
		\textbf{Input:} Gate synthesis matrix $A\in \mathbb{Z}_2^{(n,m)}$. \\
		\textbf{Output:} Gate synthesis matrix $A^\prime \in \mathbb{Z}_2^{(n,m^\prime)}$ such that $m^\prime \leq m$ and $S^{(A^\prime)}=S^{(A)}$.
		\footnotesize
		\begin{itemize}
			\item Let $\mathrm{col}_j(A)$ be a function that returns the $j^{\text{th}}$ column of $A$.
			\item Let $\text{cols}(A)$ be a function that returns the number of columns of $A$.
			\item Let $\text{nullspace}(A)$ be a function that returns a matrix whose columns generate the right null space of A.
			\item Let $\mathrm{proper(A)}$ be a function that returns matrix $A$ with every pair of identical columns and every all-zero column removed.
		\end{itemize}
		\normalsize
		\begin{algorithmic}			
			\Procedure{TODD}{}
			\State{Initialize $A^\prime \leftarrow A$}
			\BState \emph{start}:
			\ForAll{$1\leq a < b \leq \mathrm{cols}(A^\prime)$}
			\State{$\mathbf{z}\leftarrow \text{col}_a(A^\prime) + \text{col}_b(A^\prime)$}
			\State{$\tilde{A}\leftarrow \begin{pmatrix}
				A^\prime \\	
				\chi(A^\prime,\mathbf{z})
				\end{pmatrix}$}
			\State{$N\leftarrow \text{nullspace}(\tilde{A})$}			
			\ForAll{$1 \leq k \leq \text{cols}(N)$}
			\State{$\mathbf{y}\leftarrow\text{col}_k(N)$}
			\If{$y_a\oplus y_b=1$}			
			\If{$|\mathbf{y}|=1 \pmod 2$}			
			\State{$A^\prime \leftarrow \begin{pmatrix}
				A^\prime & \mathbf{0}
				\end{pmatrix}$}
			\State{$\mathbf{y}\leftarrow \begin{pmatrix}
				\mathbf{y} \\
				1
				\end{pmatrix}$}					
			\EndIf
			\State{$A^\prime \leftarrow A^\prime + \mathbf{z}\mathbf{y}^T$}
			\State{$A^\prime \leftarrow \mathrm{proper(A^\prime)}$}
			\State{\textbf{goto} \emph{start}}
			\EndIf
			\EndFor
			\EndFor
			\EndProcedure			
		\end{algorithmic}
	\end{algorithm}

	\FloatBarrier

	\section{Computational Efficiency of TODD}
	In this appendix, we calculate an upper-bound on the worst-case computational efficiency of the TODD algorithm as described in appendix~\ref{pseudocode}, in terms of the number of arithmetic operations on $GF(2)$ required.

	Let $A$ be a gate synthesis matrix with $n$ rows and $m$ columns that is used as input for the TODD algorithm.
	The loop, $L_1$, over each column pair $(a,b)$ requires at most $\binom{m}{2} = O(m^2)$ iterations to complete.
	Inside $L_1$, there are four lines of pseudocode: a column addition, requiring no more than $n$ operations; a matrix concatenation and calculation of $\chi(A,\mathbf{z})$, requiring $E_1$ operations; a nullspace calculation, requiring $O(n^3) + O(m^2n)$ operations using Gaussian elimination; and finally a nested loop $L_2$, requiring $E_2$ operations.
	
	From equation~\eqref{eq_chi_row}, we see that each row of $\chi(A, \mathbf{z})$ can be calculated with $O(m)$ operations.
	There are a maximum of ${n}\choose{3}$ rows in $\chi(A, \mathbf{z})$ so the total number of operations required to calculate $\chi$ is $O(n^3m)$.
	Combining this with the matrix concatenation, we find that $E_1 = O(n^3m) + nm = O(n^3m)$.

	The loop $L_2$ executes in at most
	\begin{equation}
	\textsc{cols}(\textsc{nullspace}(\tilde{A})) := \textsc{colrank}(\textsc{nullspace}(\tilde{A})) = m - \textsc{rank}(\tilde{A}) \leq m - \textsc{rank}(A) \leq m - n
	\end{equation} iterations.
	The identity between the column rank and the number of columns follows from the assertion that the nullspace function outputs a matrix whose columns are a linearly independent basis for the nullspace of $A$.

	The loop $L_2$ is composed of a conditional that requires 1 addition (by merging the first line of $L_2$ and the conditional).
	The content of the conditional is only evaluated once, so can be considered as part of $L_1$ for this calculation.
	Therefore, the number of operations performed in $L_2$ is $E_2 = m-n$.
	
	The nested conditional requires at most $m + n + 1$ operations, where the terms are due to the Hamming weight of $|\mathbf{y}|$,  concatenating an all-zero column to $A^\prime$ and concatenating a one to $\mathbf{y}$, respectively.
	The line $A^\prime \leftarrow A^\prime + \mathbf{z}\mathbf{y}^T$ requires at most $n(m+1)$ operations and the proper function can be computed using at most $m$ operations by keeping track of all-zero columns with a Boolean array, for a small physical overhead of $m$.

	The outermost loop (between \emph{start} and \textbf{goto} \emph{start}) by definition executes in no more than $m-m^\prime$ iterations where $m^\prime$ is the number of columns of the output. In this worst-case calculation, we assume $m^\prime = 0$.
	
	So the TODD algorithm can be executed using
	\begin{align}
	&O(m\left[n+O(n^3m) + O(n^3) + O(m^2n) + (m-n) + (m+n+1) + n(m+1) + m\right])
	\\ & = O(m\left[O(n^3m) + O(n^3)+ O(m^2n)\right])
	\\ & =O(n^3m^2) + O(nm^3)
	\end{align}
	operations.

	Therefore, given a family of Clifford + $T$ circuits with $n$ qubits, $h$ Hadamard gates and $t$ $T$ gates, we would expect our compiler to execute in time asymptotically upper-bounded by a function of the following form 
	\begin{align}
	&O((n+h)^3t^2) + O((n+h)t^3) \\
	& = O(n^3t^2) + O(h^3t^2) + O(nt^3) + O(ht^3),
	\end{align}
	where we have made the reasonable assumption that the computational bottleneck is due to the TODD algorithm, rather than the circuit preprocessing stages or mapping between different circuit representations, for instance.
	
	In practice, the actual runtimes for the benchmark quantum circuits seen in Table~\ref{tab_CliffT} are much lower than this worst-case upper-bound.
	Furthermore, the compiler runtime is dependent on the structure of the input quantum circuit, rather than simply the number of qubits and gates from which it is composed.
	Consequently, we do not see a simple relation between circuit parameters $n,t,h$ and the runtime for the benchmarks in Table~\ref{tab_CliffT}. 

	Note that in our calculation of the complexity, we assumed that we must calculate \emph{every} row of $\chi(A, \mathbf{z})$. In practice, we find that many of the rows are identical. An algorithm that calculates only the unique rows may lead to improved computational efficiency.	

	\end{appendices}

\end{document}